\pdfoutput=1
\documentclass[aps,superscriptaddress,pra,10pt,notitlepage,nofootinbib]{revtex4-1}
\usepackage{amsmath,amsfonts,amssymb,amsthm,graphicx,enumerate,bbm,longtable,xhfill,multirow}
\usepackage[caption=false]{subfig}
\usepackage{mathtools}
\usepackage[colorlinks]{hyperref}
\usepackage[normalem]{ulem}
\usepackage{natbib}

\newcommand{\ket}[1]{\left \vert #1 \right \rangle}
\newcommand{\bra}[1]{\left \langle #1 \right \vert}

\newcommand{\id}{\mathbbm{1}}

\newcommand{\cc}{\mathbbm{C}}

\DeclareMathOperator{\polylog}{polylog}
\DeclareMathOperator{\poly}{poly}

\newcommand{\ketbra}[2]{\left \vert #1 \right \rangle \! \!\left \langle #2 \right \vert}

\newcommand{\braket}[2]{\left \langle #1 \vphantom{#2} \right \vert \left. #2 \vphantom{#1} \right \rangle}

\newcommand{\be}{\begin{equation}}
\newcommand{\ee}{\end{equation}}

\newcommand{\T}{\mathcal T}
\newcommand{\U}{\mathcal U}

\newtheorem{theorem}{Theorem}
\newtheorem{lemma}{Lemma}
\newtheorem{corollary}{Corollary}
\newtheorem{proposition}{Proposition}

\theoremstyle{definition}

\DeclareMathOperator{\vspan}{span}
\DeclareMathOperator{\arccosh}{arccosh}

\begin{document}

\title{Faster ground state preparation and high-precision ground energy estimation with fewer qubits}

\author{Yimin Ge}
\author{Jordi Tura}
\author{J. Ignacio Cirac}
\affiliation{Max-Planck-Institut  f{\"u}r Quantenoptik, D-85748 Garching, Germany}

\begin{abstract}
We propose a general-purpose quantum algorithm for preparing ground states of quantum Hamiltonians from a given trial state. The algorithm is based on techniques recently developed in the context of solving the quantum linear systems problem \cite{ChildsKothariSomma15}. We show that, compared to algorithms  based on phase estimation, the runtime of our algorithm is exponentially better as a function of the allowed error, and at least quadratically better as a function of the overlap with the trial state.  We also show that our algorithm requires fewer ancilla qubits than existing algorithms,  making it attractive for early applications of small quantum computers. Additionally, it can be used to determine an unknown ground energy faster than with phase estimation if a very high precision is required. 
\end{abstract}
\maketitle

\section{Introduction}

Quantum computers are expected to have a deep impact in the simulation of large quantum systems, as originally envisioned by Feynman \cite{Feynman1982}. Of particular interest is the potential ability to study both, the dynamics  and low energy properties of many-body quantum systems, which are usually inaccessible classically due to the exponential dimension of the underlying Hilbert space. Quantum computers do not suffer from this representability problem, as one can store states in a number of qubits that only scales logarithmically with that dimension. This fact can be used to develop
very efficient algorithms to simulate the dynamics of quantum systems \cite{Lloyd96,BCCKS15,BerryChildsKothari15,LowChuangSP,LowChuang16}. However, preparing certain physically relevant states, like the ground state of a many-body Hamiltonian,  may  be significantly more difficult. This can be seen as a consequence of ground state preparation likely being hard in full generality, as indeed many variations of ground state energy problems have been proven to be complete for the class QMA \cite{QuantumHamiltonianComplexity}. Nevertheless, preparing ground states of Hamiltonians has profound applications in several fields of science so that more efficient quantum algorithms than the ones existing \cite{Kitaev95,PoulinWocjan09,AbramsLloyd99} are highly desired. This could allow one, for instance, to prepare the initial states that are required to simulate quenches in quantum many-body systems, thus enabling the study of many intriguing and not fully understood phenomena, such as many-body localisation \cite{ANDP:ANDP201700169} or the presence of thermalisation in closed systems \cite{RevModPhys.83.863,*doi:10.1080/00018732.2016.1198134}, with quantum comptuters. Other applications include single-copy tomography \cite{PhysRevLett.105.190503} and the construction of QMA witnesses \cite{PoulinWocjan09}. Furthermore, the ability to determine the ground energy of a Hamiltonian to a high precision also possesses many applications in the fields of physics and quantum chemistry \cite{QChemIntro}, and possibly even in quantum machine learning \cite{Biamonte2017}.

Most existing quantum algorithms for ground state preparation are based on one of two methods. First, one could naively attempt to project a trial state $\ket\phi$ onto the ground state by measuring the energy of $\ket\phi$ using the phase estimation algorithm \cite{Kitaev95}. 
The probabily of success is proportional to $|\phi_0|^2$, where $\phi_0$ is the overlap of $\ket\phi$ with the ground state.  
Furthermore, straightforward application of phase estimation becomes expensive if a very high fidelity of the prepared state with the real ground state is required. 
A second class of algorithms is based on variants of the adiabatic algorithm \cite{Farhi2000}. Here, the target Hamiltonian $H(1)$ is connected to a trivial Hamiltonian $H(0)$ via a path $H(s)$, which is slowly changed from $H(0)$ to $H(1)$. The adiabatic theorem guarantees that if the initial state is  the ground state of $H(0)$, which is assumed to be easily prepared, then for sufficiently long runtimes, the final state will be close to the ground state of $H(1)$. Rigorous bounds \cite{Jansen07} on the runtime however depend inverse polynomially on the minimum spectral gap along the entire path $H(s)$, which is generally exponentially small and moreover extremely difficult to calculate or bound in practice. Thus, adiabatic algorithms are often only employed as \emph{heuristic} methods to first obtain a state with (hopefully) non-trivial overlap with the ground state, which can then subsequently be used as the trial state in phase estimation \cite{Oh08}. This approach is expected to work significantly better than just using random trial states and is the current paradigm e.g. for quantum chemistry applications \cite{QChemIntro}.

In this paper, we propose a quantum algorithm that significantly improves the part played by phase estimation in this approach. More generally, we consider the problem of preparing a good approximation of the ground state from a given trial state. 
We show that compared to using phase estimation, the runtime of our ground state preparation algorithm scales exponentially better in the allowed error to the real ground state, and polynomially better with the spectral gap and the  overlap of the trial state with the ground state. We also show that, in case the ground energy is not known beforehand, the same algorithm can be used to obtain a good estimate of the ground energy to a high precision faster than is possible with phase estimation.

Unlike algorithms based on the adiabatic theorem, whose runtimes always depend on the minimum spectral gap along an entire path of Hamiltonians, all algorithms analysed in this paper only require a lower bound on the spectral gap of the target Hamiltonian. This is a significantly weaker assumption and indeed, for many systems of interest such as in typical critical points, this gap is known to scale only inverse polynomially with the system size.

The outline of the remainder of the paper is as follows. In Section~\ref{sec:overview}, we give an overview of the results and a high-level overview of the ideas. In Section~\ref{sec:known}, we give the technical details of the algorithm in case the ground energy is known beforehand. In Section~\ref{sec:unknown}, we present the technical details of the algorithm for both, ground state preparation and high-precision ground energy estimation, in case the ground energy is unknown beforehand. We close the main part of the paper with some concluding remarks and open questions in Section~\ref{sec:conclusion}. 
Appendix~\ref{app:PEAenergy} analyses the cost of finding the ground energy with phase estimation. 
In Appendix~\ref{app:PEApreparation}, we demonstrate that if phase estimation is used for ground state preparation, an extremely  precise estimate of the ground energy is required beforehand, and analyse the cost of doing so. In Appendix~\ref{app:PoulinWocjan}, we analyse the filtering algorithm from \cite{PoulinWocjan09}. Finally, in Appendix~\ref{app:Chebyshev}, we sketch an alternative approach (inspired by  the ``Chebyshev method'' of \cite{ChildsKothariSomma15}) to the problem.

\section{Overview of results}\label{sec:overview}

Throughout this paper, let $\tilde H$ be an  $N\times N$ Hermitian matrix  
such that its spectrum is contained in $[0,1]$. 
We assume that we are given the ability to efficiently perform Hamiltonian simulation of $\tilde H$. More precisely, we require that $e^{\pm i\tilde H}$ can be approximated to error $\epsilon'$ using $O(\Lambda t\polylog(N,1/\epsilon'))$ elementary gates\footnote{Examples of such algorithms include \cite{LowChuang16,BerryChildsKothari15,BCCKS15,LowChuangSP}. Notice that algorithms based on Trotter product formulas such as \cite{Lloyd96} do not meet this requirement.}, where $\Lambda$ is the ``base cost'' of the simulation (e.g., if the simulation algorithm works in the oracle model \cite{BerryChildsKothari15}, $\Lambda$ is the gate cost of the oracles). 
Let $\lambda_0$ be the lowest eigenvalue of $\tilde H$, and $\ket{\lambda_0}$ be the corresponding eigenstate. For simplicity of notation, we will assume that $\lambda_0$ is non-degenerate (all results in this paper trivially generalise to the case when $\lambda_0$ is degenerate, see Section~\ref{sec:conclusion}). Suppose that $\Delta$ is a known lower bound on the spectral gap of $\tilde H$.

Suppose moreover that we are given a circuit $\mathcal C_\phi$ using $\Phi$ elementary gates that prepares a trial state $\ket\phi$. Let $\phi_0:=\braket{\lambda_0}\phi$ be its (generally unknown) overlap with the ground state, and $\chi$ be a known lower bound on $|\phi_0|$. We will throughout this paper assume that $\chi = e^{-O(\log N)}$. Notice that this is an extremely weak assumption, indeed, this is satisfied even for random states with high probability.  
The aim of this paper is to prepare a state $\epsilon$-close to $\ket{\lambda_0}$ by (approximately) projecting $\ket\phi$ onto its ground state component. 

Throughout this paper, we will use the computer science convention for the big-$O$ notation. We furthermore use $\tilde O$ to denote the complexity up to polylogarithmic factors in $N$, $\Delta^{-1}$, $\epsilon^{-1}$, $|\phi_0|^{-1}$ and $\chi^{-1}$. 
Our first result can now be stated as follows. 

\begin{theorem}[Ground state preparation for known ground energy] \label{thm:GSEknown}
 Suppose that $\lambda_0$ is known to a precision of $O(\Delta/\log \frac1{\chi \epsilon})$. Then, an $\epsilon$-close state to $\ket{\lambda_0}$ can be prepared with constant probability in a gate complexity of 
 \begin{equation} \label{eq:gatesknown}
\tilde O\left(\frac\Lambda{|\phi_0| \Delta}+\frac {\Phi}{|\phi_0|}\right),
\end{equation}
and using
\begin{equation} \label{eq:qubits}
	O\left(\log N+\log\log\frac1{\epsilon}+\log \frac1\Delta\right)
\end{equation}
qubits. Moreover, a flag qubit indicates success. 
\end{theorem}

\begin{table*}[ht]
 \begin{tabular}{l|c|c|c}
 Preparation (ground energy known) & Gates & Qubits& Required precision  \\
 \hline
 This paper 
 & $\displaystyle \tilde O\left(\frac\Lambda{|\phi_0|\Delta}+\frac{\Phi}{|\phi_0|}\right)$ & 
 {$\displaystyle O\left(\log N+\log\log\frac1{\epsilon}+\log\frac1\Delta \right)$}& 
	$\tilde O(\Delta)$ \\
 Phase estimation + amp. amplif. & $\displaystyle\tilde O\left(\frac\Lambda{|\phi_0|^2\Delta\epsilon}+\frac{\Phi}{|\phi_0|}\right)$ & $\displaystyle O\left(\log N+\log \frac1{\epsilon}+\log\frac1\Delta\right)$ & $O\left({|\phi_0|\epsilon\Delta}\right)$\\
Filtering (Poulin  \& Wocjan) 
 & $\displaystyle\tilde O\left(\frac\Lambda{|\phi_0|\Delta}+\frac{\Phi}{|\phi_0|}\right)$ & $\displaystyle O\left(\log N+\log \frac1{\epsilon}+\frac{\log\frac1{\chi\epsilon}}{\log\log\frac{1}{\chi\epsilon}}\times\log\frac 1\Delta\right)$ & 
$\tilde O(\Delta)$  
\end{tabular}
\caption{Algorithms for ground state preparation for the case when the ground energy is known beforehand to the required precision.} \label{tab:prepareCostknown}
\end{table*}

Although other quantum algorithms for this or similar purposes have previously been proposed \cite{Kitaev95,PoulinWocjan09,AbramsLloyd99}, to the best of our knowledge, the algorithm in this paper, for the case of known ground  energy, exhibits  the best scaling for both the runtime and the number of qubits amongst all existing algorithms 
so far (see Table~\ref{tab:prepareCostknown}). 
For example, the common approach of combining phase estimation with amplitude amplification \cite{BHMT02} has a runtime that is exponentially worse in $\epsilon$,  and moreover quadratically worse in $|\phi_0|$  (see Appendix~\ref{app:PEApreparation}). In fact, an inverse polynomial dependence on $\epsilon$ is common to almost all algorithms that are based on phase estimation \cite{Kitaev95,AbramsLloyd99}. To the best of our knowledge, the only exception is a filtering method proposed by Poulin and Wocjan \cite{PoulinWocjan09}, which was originally designed to quadratically improve the runtime dependence on $|\phi_0|$ to obtain a state with low expected energy, and which, as we prove in Appendix~\ref{app:PoulinWocjan}, can also be used to obtain the ground state with a runtime scaling that is  polylogarithmic in $\epsilon^{-1}$ with a suitable choice of parameters\footnote{The filtering method can also be formulated as a majority voting scheme \cite{ChildsKothariSomma15,KothariCorrespondence}.}. 
 This however comes at the cost of requiring significantly more ancilla qubits, which makes it challenging for early applications of small quantum computers. 
 
\begin{table*}[ht]
\qquad(a) \hfill 
 \begin{tabular}[t]{l|c|c}
Preparation (ground energy unknown)& Gates & Qubits \\
 \hline
  This paper & $\displaystyle \tilde O\left(\frac\Lambda{\chi\Delta^{3/2}}+\frac{\Phi}{\chi\sqrt\Delta}\right)$ & 
 $\displaystyle O\left(\log N+\log\log\frac1{\epsilon}+\log\frac1\Delta \right)$ \\
 Phase estimation + min. label finding &  $\displaystyle\tilde O\left(\frac\Lambda{\chi^4\Delta\epsilon}+\frac{\Phi}{\chi}\right)$  & $\displaystyle O\left(\log N+\log \frac1{\epsilon}+\log\frac1\Delta\right)$  \\
 Filtering + min. label finding &  $\displaystyle\tilde O\left(\frac\Lambda{\chi\Delta^{3/2}}+\frac{\Phi}{\chi\sqrt\Delta}\right)$ & $\displaystyle O\left(\log N+\log \frac1{\epsilon}+\frac{\log\frac1{\chi\epsilon}}{\log\log\frac{1}{\chi\epsilon}}\times\log\frac 1\Delta\right)$\\ 
  \emph{Combined approaches} 
   & & \\
 This paper + phase estimation &  $\displaystyle \tilde O\left(\frac{\Lambda}{\chi^3\Delta}+\frac{\Phi}{\chi}\right)$ & 
 $\displaystyle O\left(\log N+\log \frac1{\epsilon}+\log\frac1\Delta\right)$  \\
 Filtering  + phase estimation&  $\displaystyle\tilde O\left(\frac{\Lambda}{\chi^3\Delta}+\frac{\Phi}{\chi}\right)$ & $\displaystyle O\left(\log N+\log \frac1{\epsilon}+\frac{\log\frac1{\chi\epsilon}}{\log\log\frac{1}{\chi\epsilon}}\times\log\frac 1\Delta\right)$ 
\end{tabular} \qquad \hspace{2pt}
\vspace{10pt}

\qquad(b) \hfill  \begin{tabular}[t]{l|c|c}
Ground energy estimation & Gates & Qubits \\
 \hline
 This paper  &$\displaystyle \tilde O\left(\frac\Lambda{\chi\xi^{3/2}}+\frac{\Phi}{\chi\sqrt\xi}\right)$ & 
 {$\displaystyle O\left(\log N+\log\frac1\xi \right)$} \\
 Phase estimation + min. label finding &  $\displaystyle\tilde O\left(\frac\Lambda{\chi^3\xi}+\frac{\Phi}{\chi}\right)$ & $\displaystyle O\left(\log N+\log\frac1\xi \right)$ 
	\\
Filtering + min. label finding  & $\displaystyle \tilde O\left( \frac{\Lambda}{\chi\xi^{3/2}}+ \frac{\Phi}{\chi\sqrt\xi} \right)$ & $\displaystyle O\left(\log N+\frac{\log\frac1{\chi}}{\log\log\frac{1}{\chi}}\times\log\frac 1\xi\right)$  \\
 \emph{Combined approaches} 
   & & \\
This paper + phase estimation &  $\displaystyle \tilde O\left(\frac{\Lambda}{\chi^3\xi}+\frac{\Phi}{\chi}\right)$ & 
 $\displaystyle O\left(\log N+\log\frac1\xi\right)$  \\
Filtering  + phase estimation&  $\displaystyle\tilde O\left(\frac{\Lambda}{\chi^3\xi}+\frac{\Phi}{\chi}\right)$ & $\displaystyle O\left(\log N+\frac{\log\frac1{\chi}}{\log\log\frac{1}{\chi}}\times\log\frac 1\xi\right)$ 
\end{tabular}\hfill \qquad\qquad\qquad

\caption{Algorithms in case the ground energy is not known beforehand. (a) Algorithms for ground state preparation. (b) Algorithms for estimating the ground energy to a precision of $\xi\ll \Delta$. 
The minimum label finding algorithm is a subroutine that we describe in Section~\ref{subsec:MinFinding}.  The combined approaches have been adjusted to yield the optimal scaling in $\Delta$ and $\xi$, respectively.} 
\label{tab:CostUnknown}
\end{table*}

The algorithm can also be adapted for the case when the ground energy is not known beforehand. 

\begin{theorem}[Ground state preparation for unknown ground energy] \label{thm:GSEunknown}
	If the ground energy is not known beforehand, the same task as in Theorem~\ref{thm:GSEknown}	
	can be achieved in a gate complexity of 
	\begin{equation} \label{eq:gatesunknown}
		\tilde O \left(\frac\Lambda{\chi \Delta^{3/2}}+\frac{\Phi}{\chi\sqrt\Delta} \right)
	\end{equation}
	and the same number \eqref{eq:qubits} of qubits.
\end{theorem}
 Provided $\Phi$ is not too large (which can be assumed in most practical scenarios), our algorithm for the case of unknown ground energy also has a better runtime scaling than naive phase estimation, and uses significantly fewer qubits than an adaption of Poulin \& Wocjan's filtering method \cite{PoulinWocjan09} for this task (see Table~\ref{tab:CostUnknown}a). 

Furthermore, for very small $\Delta$, we show that alternatively, the scaling in $\Delta$ can be improved to $\sim \Delta$ at the expense of worsening the scaling in $\chi$ by combining our algorithm with a prior run of phase estimation to first obtain an estimate of the ground energy. 

\begin{theorem}[Combined algorithm for ground state preparation] \label{thm:combined}
	The same task as in Theorem~\ref{thm:GSEunknown} can be achieved in a gate complexity of 
	\begin{equation} \label{eq:gatescombined}
 	\tilde O\left( \frac{\Lambda}{\chi^3\Delta^\kappa} + \frac{\Lambda}{\chi \Delta^{(3-\kappa)/2}} + \frac{\Phi}{\chi\Delta ^{(1-\kappa)/2}} \right)
	\end{equation}
	for any choice of $\kappa\in[0,1]$, and the same number \eqref{eq:qubits} of qubits. 
\end{theorem}
In particular, choosing $\kappa=1$ in Theorem~\ref{thm:combined} yields the optimal scaling in $\Delta$ of $\sim\Delta$.

We show moreover that, with high probability, the algorithms for the case of unknown ground energy also find  the ground energy to a precision of $\tilde O(\Delta)$. Since $\Delta$ can be any reliable lower bound on the spectral gap, this yields a general algorithm for estimating the energy. 

\begin{theorem}[High-precision ground energy estimation] \label{thm:estimation}
	Let $\xi = \tilde O(\Delta)$. Then $\lambda_0$ can be estimated to an additive precision of $\xi$ in a gate complexity of 
	\begin{equation}
		\tilde O\left(\frac\Lambda{\chi\xi^{3/2}}+\frac{\Phi}{\chi\sqrt\xi}\right),
	\end{equation}
	using $O(\log N+ \log \xi^{-1})$ qubits. Alternatively, the combined approach of Theorem~\ref{thm:combined} achieves this task in a gate complexity of 
	\begin{equation} 
 	\tilde O\left( \frac{\Lambda}{\chi^3\xi^\kappa} + \frac{\Lambda}{\chi \xi^{(3-\kappa)/2}} + \frac{\Phi}{\chi\xi^{(1-\kappa)/2}} \right)
	\end{equation}
	and the same number of qubits. 
\end{theorem}
Provided that $\Phi$ is not too large, this also scales better than performing the same task with phase estimation and amplitude amplification (see Table~\ref{tab:CostUnknown}b).

In terms of query complexities, the coefficients of  $\Phi$ and $\Lambda$  in Tables~\ref{tab:prepareCostknown}--\ref{tab:CostUnknown}  are, up to polylogarithmic factors, the number of calls to  $\mathcal C_\phi$ and (unit time) Hamiltonian simulation, respectively. Note that the qubit requirements stated in Tables~\ref{tab:prepareCostknown}--\ref{tab:CostUnknown} exclude ancilla qubits required to preform Hamiltonian simulation\footnote{In fact, the Hamiltonian simulation algorithms  \cite{LowChuang16,BerryChildsKothari15,LowChuangSP} only require few ancilla qubits and leave the $O(\cdot)$ expressions in Tables~\ref{tab:prepareCostknown}--\ref{tab:CostUnknown} unchanged. \cite{BCCKS15} requires  $O\left(\frac{\log (d)\log(\beta \Delta^{-1}\log(\chi^{-1})/\epsilon)}{\log\log(\beta \Delta^{-1}\log(\chi^{-1})/\epsilon)}\right)$ additional qubits, where $\tilde H = \sum_{j=1}^d \beta_j U_j$ with unitaries $U_j$ costing $O(\Lambda)$ elementary gates and $\beta = \sum_j |\beta_j|$ (see \cite[Table 1]{LowChuang16} for an overview).}. 

Our algorithms are inspired by classical power iteration methods \cite{ZAMM}.
They are based on techniques (termed the ``Fourier method'') that were recently developed for the quantum linear systems problem \cite{ChildsKothariSomma15}, and are based on the observation that the \emph{Linear Combination of Unitaries}, or \emph{LCU Lemma} \cite{BerryChildsKothari15}, can be used to implement other functions of a Hamiltonian \cite{arXiv:1705.01843}.

We now briefly outline the basic idea of the algorithms. It is easy to see that for  positive-semidefinite $H$ with non-degenerate lowest eigenvalue $0$, high powers of $\cos H$ approximately project any given state into a state proportional to the unique ground state of $H$. In case the ground energy of $\tilde H$ is known, $\tilde H$ can be easily transformed into another Hamiltonian $H$ such that $\ket{\lambda_0}$ is the unique eigenvector of $H$ of eigenvalue $\approx 0$. The outline of the algorithm in that case is as follows:
\begin{enumerate}\itemsep1pt 
\item Approximate  $\cos^M H$ as a linear combination of terms of the form $e^{-iHt_k}$. 
\item Using the techniques in \cite{ChildsKothariSomma15}, we implement this linear combination with some amplitude using Hamiltonian simulation and the  LCU Lemma. 
\item We use amplitude amplification to boost the overlap with the target state. Alternatively, the fixed point search algorithm \cite{YoderLowChuang14} can be used for this step. 
\end{enumerate}

The outline of the algorithm in case the ground energy is unknown beforehand is as follows:
\begin{enumerate}
\item Adapt the previous algorithm into a circuit controlled on an ancilla register $\ket E$, which runs steps 1 and 2 of the previous algorithm, assuming that the ground energy were $E$. 
\item Divide $[0,1]$ into $L=\tilde O(\Delta^{-1})$ equally spaced values $E_0,\ldots,E_{L-1}$. Run the algorithm with $\sqrt{L^{-1}}\sum \ket {E_j}$ on the ancilla register. 
\item Use the minimum label finding algorithm (Section~\ref{subsec:MinFinding}) to search for the smallest $j$ such the residual state of $\ket{E_j}$ has large norm.
\item When this search succeeds, then with high probability the resulting $E_j$ is within $\xi$ of the true ground energy and the residual state is a good approximation of the ground state.
\end{enumerate}

The outline of the combined approaches is as follows:
\begin{enumerate}
\item Use phase estimation and amplitude amplification to obtain a ``crude'' estimate of the ground energy. This provides an interval $I$ which is known to contain the real ground energy. 
\item Take $L\approx |I|/\xi$ equally spaced values $E_0,E_1,\ldots E_{L-1}$ in $I$, and 
run the previous algorithm with these values of $E_j$. 
\end{enumerate}

\section{Algorithm for the case of known ground energy} \label{sec:known}

In this section, we present the main technical analysis of  our algorithm and prove Theorem~\ref{thm:GSEknown}. 
Suppose that the value of $\lambda_0$ is known to a precision of $\delta=O\left(\Delta/\log\frac1{\chi\epsilon}\right)$. 
Let $E$ be a known value such that $0\leq E\leq\lambda_0$ and $\delta_E:=\lambda_0-E <\delta$. 
Define $H:=\tilde H - (E-\tau)\id$ for some small value of $\tau$ chosen below. Then, $\ket{\lambda_0}$ is the (unique) ground state of $H$ with eigenvalue $\tau + \delta_E$, and by assumption all other eigenvalues of $H$ are $ \geq \tau+\delta_E + \Delta$. 
The method presented here is based on the observation that high powers of $\cos  H$ are approximately proportional to projectors onto the ground state:  
\begin{equation}
	\cos^M  H \ket\phi = \phi_0\cos^M(\tau+\delta_E)\left( \ket{\lambda_0} + \frac1{\phi_0}\frac{ \cos^M  H}{\cos^M(\tau+\delta_E)}\ket{\lambda_0^\perp}\right).
\end{equation}
The norm of the second term is bounded by $|\phi_0|^{-1} e^{-O(M(\tau+\delta_E)\Delta)}$. Indeed, since $\cos x$ is concave and decreasing on $[\tau, 1+\tau]$, 
\begin{align}
 \left\|\frac{ \cos^M   H}{\cos^M(\tau+\delta_E )}\ket{\lambda_0^\perp}\right\| &<  \left(\frac{\cos(\tau+\delta_E)-\sin(\tau+\delta_E)\Delta}{\cos(\tau+\delta_E)}\right)^M \\
 &= \left( 1- \tan(\tau+\delta_E)\Delta\right)^M \\
 &= e^{-\Omega(M\tan(\tau+\delta_E)\Delta)} = e^{-\Omega(M(\tau+\delta_E)\Delta)},
\end{align}
where in the last step we used $\tan x \geq x$ for $x\in[0,1+\tau]$. 
Thus, 
\begin{equation} \label{eq:cosFullApprox}
	\left\| \frac{\cos^M H \ket\phi}{\|\cos^M  H \ket\phi\| } - \ket{\lambda_0} \right\| =O(\epsilon),
\end{equation}
provided that  
\begin{equation}
M=\Omega\left(\frac1{\Delta(\tau+\delta_E)} {\log\frac1{|\phi_0|\epsilon}}\right).
\end{equation}
On the other hand, using that $\cos x > 1 - x^2/2$, 
\begin{align}
	\cos^M(\tau+\delta_E) &> \left(1-\frac{(\tau+\delta_E)^2}2\right)^M\\
	&= e^{-O((\tau+\delta_E)^2M)}.
\end{align}
Thus,
\begin{equation} \label{eq:cosNorm}
\|\cos^M H\ket\phi\| = \Omega(|\phi_0|),
\end{equation}
provided that $\tau+\delta_E = O(1/\sqrt M)$. Hence, since by assumption $\delta_E < \delta$, choosing
\begin{equation} \label{eq:choicetau}
	\tau = \Theta\left( \frac{\Delta}{\log\frac1{\chi\varepsilon} }\right) 
\end{equation}
and 
\begin{equation}
 M = \Theta\left( \frac1{\Delta^2} \log^2\frac1{\chi\epsilon}\right)
\end{equation}
satisfies both \eqref{eq:cosFullApprox} and \eqref{eq:cosNorm}. 

Our aim in the following is to prepare $\cos^M  H\ket\phi$. The strategy we employ is as follows. First, we approximate $\cos^M  H$ as a linear combination of few unitaries of the form $e^{-i Ht_k}$ for suitable values of $t_k$. Second, we implement this linear combination with some amplitude using Hamiltonian simulation and the LCU Lemma \cite{BerryChildsKothari15}. Third, we use amplitude amplification or fixed point search to boost the overlap with the target state. 

In the following, assume for simplicity that $M=2m$ is even (the algorithm can be adapted to odd $M$ with minor modifications). Observe that
\begin{equation}
 \cos^{2m}x = \left(\frac{e^{ix} + e^{-ix }}2 \right)^{2m} 
  = 2^{-2m} \sum_{k=-m}^m \binom{2m}{m+k} e^{2ikx}.
\end{equation}
Note that
\begin{equation}\label{eq:chernoff}
 2^{-2m} \sum_{k=m_0+1}^m \binom{2m}{m+k} \leq e^{-m_0^2/4m}.
\end{equation}
Indeed, the LHS is the probability of seeing more than $m+m_0$ heads when flipping $2m$ coins, and \eqref{eq:chernoff} follows from the Chernoff bound. Thus, 
\begin{equation}\label{eq:truncFourier}
\cos^{2m} H =  \sum_{k=-m_0}^{m_0} \alpha_k e^{-2i H k} + O(\chi\epsilon), 
\end{equation}
where
\begin{equation}\label{eq:alphak}
\alpha_k := 2^{-2m} \binom{2m}{m+k},
\end{equation}
and 
\begin{equation} \label{eq:m0Fourier}
m_0 = \Theta\left(\sqrt{M\log\frac1{\chi\epsilon}}\right)
=  \Theta\left(\frac1{ \Delta}\log^{3/2}\frac1{\chi\epsilon}\right).
\end{equation}

Next,  $e^{-2i H k}$ can be implemented using Hamiltonian simulation algorithms. 
To implement the RHS of \eqref{eq:truncFourier}, we employ the LCU Lemma: Let $B$ be a circuit on $b:=\lceil\log_2(2m_0+1)\rceil$ qubits that maps $\ket 0^{\otimes b} $ to 
\begin{equation}
B\ket 0^{\otimes b} := \frac1{\sqrt\alpha} \sum_{k=-m_0}^{m_0} \sqrt{\alpha_k}\ket k,
\end{equation}
where $\alpha = \sum_{k=-m_0}^{m_0}\alpha_k$, and let $U$ be the controlled Hamiltonian simulation $U\ket k\ket\phi = \ket k e^{-2i Hk}\ket\phi$. Then,
\begin{equation}
 (B^\dagger\otimes \id) U(B\otimes\id) \ket\phi   = \frac1\alpha \ket0^{\otimes b} \sum_{k=-m_0}^{m_0} \alpha_k e^{-2i Hk}\ket\phi + \ket{R}
\end{equation} 
where $(\ketbra00^{\otimes b}\otimes\id)\ket R=0$. 

The final step of the algorithm is to boost the overlap with amplitude amplification or fixed point search. Measuring the ancillas will then project the state onto 
\begin{equation}
\ket{\lambda_0'}:=\frac{\sum_{k=-m_0}^{m_0} \alpha_k e^{-2i Hk}\ket\phi}{\|\sum_{k=-m_0}^{m_0} \alpha_k e^{-2i Hk}\ket\phi\|}
\end{equation}
with probability close to $1$.  From \eqref{eq:truncFourier},
\begin{equation}
\frac{\sum_{k=-m_0}^{m_0} \alpha_k e^{-2i Hk}\ket\phi}{\|\sum_{k=-m_0}^{m_0} \alpha_k e^{-2i Hk}\ket\phi\|} = \frac{\cos^{2m} H\ket\phi}{\|\cos^{2m} H\ket\phi\|} + O(\epsilon)
\end{equation}
Thus, \eqref{eq:cosFullApprox} implies
\begin{equation}
	\ket{\lambda_0'} = \ket{\lambda_0}+O(\epsilon),
\end{equation}
as required. Eq.~\eqref{eq:cosNorm} implies that the number of repetitions is 
$O\left(\alpha/\|\sum_{k=-m_0}^{m_0} \alpha_k e^{-2i Hk}\ket\phi\| \right)= O\left(\alpha/|\phi_0|\right)$.

We now calculate the gate count of the entire algorithm. First note that $B$ 
can be implemented with $O(2^b)=O(m_0)$ elementary gates \cite{SBM06}. Next, the gate cost to implement $e^{\pm 2iH}$ to accuracy $\epsilon'$ in operator norm is $O(\Lambda\polylog(N,\frac1{\epsilon'}))$, depending on the precise model and Hamiltonian simulation algorithm used (see \cite[Table 1]{LowChuang16} for an overview). 
Here, we require $\epsilon' = O( \epsilon|\phi_0|/m_0)$. 
Thus, the gate cost of $U$ is $O(m_0 \Lambda\polylog(N,\frac{m_0}{\epsilon|\phi_0|}))$ \cite[Lemma 8]{ChildsKothariSomma15}. 
Note that Hamiltonian simulation with respect to $H$ can be trivially obtained from Hamiltonian simulation with respect to $\tilde H$, either by a phase shift or absorbing these phases into the values of the $\alpha_k$. 
Finally, note that $\alpha=O(1)$ and each iteration of amplitude amplification or fixed point search requires $O(1)$ uses of $\mathcal C_\phi$, $B$ and $U$. The final gate complexity is thus
\begin{equation}\label{eq:finalgatecomplexityFourierKnown}
O\left(\frac{1}{|\phi_0|}\left(m_0 \Lambda\polylog\left(N,\frac{m_0}{\epsilon|\phi_0|}\right)+\Phi\right)\right) 
=O\left(\frac\Lambda{|\phi_0| \Delta}\polylog\left(N,\frac1\Delta,\frac1{|\phi_0|\epsilon}\right)+\frac {\Phi}{|\phi_0|}\right).
\end{equation}
Note that if fixed point search is used for the final step, we also require a good lower bound of $|\phi_0|$. However, we can simply run the algorithm with $O(1/\chi')$ repetitions in fixed point search for $\chi'=1,\frac12,\frac14, \ldots$ until we successfully project the ancillas into $\ket 0^{\otimes (b+q)}$. Indeed, this only results in an overall multiplicative overhead of $O(\log 1/|\phi_0|)$, yielding \eqref{eq:gatesknown}. Moreover, whenever we succeed, the resulting state is $\ket{\lambda_0'}$ independently of the value of $\chi'$ that was used. 
The number of ancilla qubits required is 
$b$ plus any ancillas necessary for performing Hamiltonian simulation. Note that the number of ancilla qubits required for implementing $U$ is essentially the same as the number of ancilla qubits required for a single run of $e^{\pm 2i\tilde H}$, since the latter can be re-used.  Thus, the total number of qubits required is given by \eqref{eq:qubits}, plus any ancilla qubits  required for Hamiltonian simulation (see also \cite[Table 1]{LowChuang16}). This proves Thm.~\ref{thm:GSEknown}. \qed

\section{Algorithm for the case of unknown ground energy} \label{sec:unknown}

The algorithm of Theorem~\ref{thm:GSEknown} presented in the previous section requires an estimate $E$ of $\lambda_0$ to a precision of $O(1/\sqrt M)$. Since $E$ can simply be viewed as an input parameter, the algorithm can in principle also be run with different values of $E$. It is easy to see that for any $E\in[0,\lambda_0]$, the algorithm would, if successful, produce a good approximation of $\ket{\lambda_0}$, but may have an exponentially small probability of success. Thus, if the ground energy is not known beforehand, one could simply run the algorithm for increasing values of $E$, using a step size $O(1/\sqrt M)$, and stop when first successful. It is moreover easy to see that the value of $E$ at which we first succeed is with high probability a good estimate of $\lambda_0$. Clearly, the runtime of this algorithm would result in an overall factor of $O(\sqrt M)=\tilde O(1/\Delta)$ compared to Eq.~\eqref{eq:gatesknown}. 

It turns out that this ``classical search'' for the correct value of $E$ can be replaced by a ``quantum search'' that quadratically improves the overhead from $O(\sqrt M)$ to $\tilde O(\sqrt[4]{M})$. We call this search the \emph{minimum label finding} algorithm, which we describe in Section~\ref{subsec:MinFinding} as a general subroutine, and which may be of independent interest. We then apply this algorithm to the ground state preparation and ground energy estimation problem in Section~\ref{subsec:proofsUnknown}. 

\subsection{Minimum label finding} \label{subsec:MinFinding}

In this section, we describe a general subroutine to find the minimum ``label'' in an ancilla register amongst terms  with at least some given amplitude in a given superposition. To motivate this result, consider the following scenario. Suppose we have $L$ unitaries $U_0,\ldots U_{L-1}$ that prepare the states $U_i\ket 0\!\ket 0 = \ket 0\! \ket{\Phi_i} + \ket{R_i}$, where $(\ketbra00 \otimes\id) \ket{R_i}=0$. Let $\chi\in(0,1)$ be a known real number. Suppose we want to approximately find the smallest $i$ (or alternatively, prepare the corresponding $\ket{\Phi_i}$) such that $\|\ket{\Phi_i}\| \geq \chi$. The naive way to do this is to use amplitude estimation (or fixed point search) to increase the amplitude of $\ket 0$ on the first register of $U_i\ket 0\!\ket 0$ for $i=0,1,\ldots$, each time using only $O(1/\chi)$ repetitions of $U_i$, until we first succeed. This requires $\tilde O(L/\chi)$ calls to the individual  $U_i$'s in total. Below, we show that a quadratic improvement in $L$ to $\tilde O(\sqrt L / \chi)$ can be achieved, provided that performing the controlled version $U=\sum_i \ketbra ii\otimes U_i$ can be done with essentially the same cost as the individual $U_i$'s. The algorithm is based on a simple binary search technique to successively find the binary digits of the smallest ``label'' $i$ amongst the terms in the superposition $ \sqrt{L^{-1}}\ket0\sum_i\ket i\ket{\Phi_i} + \ket R$ with a norm of at least $\chi/\sqrt L$. 

\begin{proposition}[Minimum label finding] \label{prop:minFinding}
Let $\mathcal C_\Phi$ be a circuit on $q+n+m$ qubits that prepares the state
\begin{equation}
	\ket\Phi := \mathcal C_\Phi\ket{0}^{\otimes(q+n+m)}=\ket 0^{\otimes q}\sum_{i=0}^{2^n-1} \ket i \ket{\Phi_i} + \ket R,
\end{equation}
where $\ket{\Phi_i}$ are non-normalised $m$-qubit states and $\ket R$ has no overlap with $\ket 0^{\otimes q}$ on the first $q$ qubits. Let $\zeta\in (0,1)$ be a known real number, and $\tilde J, J\in\{0,\ldots,2^n-1\}$ be (unknown) integers such that $\|\ket{\Phi_J}\|\geq \zeta$ and 
\begin{equation}\label{eq:AppALem1Sum}
	\sum_{i=0}^{\tilde J} \|\ket{\Phi_i}\|^2 < \zeta^2\frac{\delta}{4n\log\frac n\delta\ln^2\frac2\delta},
\end{equation}
where $\delta\in (0,1/5)$. Then, there exists a quantum algorithm that prepares the state $\ket j\ket{\Phi_j}/\|\ket{\Phi_j}\|$ for some $j\in[\tilde J,J]$ with probability at least $1-5\delta$ 
using
\begin{equation} \label{eq:AppALem1Gates}
	O\left(\frac{n}{\zeta} \log^2\frac n\delta\right)
\end{equation}
calls to $\mathcal C_\Phi$, 
\begin{equation}
O\left(\frac{\poly(q,n,m)}{\zeta} \log^2\frac n\delta\right)
\end{equation} 
additional elementary gates, and $q+n+m+1$ qubits. 
\end{proposition}

The minimum label finding algorithm uses the fixed-point search algorithm \cite{YoderLowChuang14}, which can be thought of as a variation of amplitude amplification, with the additional features that it is not possible ``overshoot'' the target state, and that only a lower bound on the initial overlap is required to be known. More precisely, let $\mathcal C$ be a circuit on $n'$ qubits that prepares the state $\mathcal C\ket 0^{\otimes n'} =\lambda\ket{T}+ \sqrt{1-\lambda^2}\ket{\bar T}$ with $\braket{T}{\!\bar T}=0$, and $U$ be a unitary that that satisfies $U\ket T\!\ket b = \ket T\! \ket{1-b}$ and $U\ket{\bar T}\!\ket b = \ket{\bar T}\!\ket b$ for $b\in\{0,1\}$. Then, given input parameters $\lambda', \delta\in(0,1)$, the fixed point search $\mathrm{FPS}(\mathcal C, U, \lambda',\delta)$ is a circuit on $n'+1$ qubits using $O(\log(1/\delta)/\lambda')$ calls to $\mathcal C, \mathcal C^\dagger, U$ and $O(n'^2\log(1/\delta)/\lambda')$ elementary gates such that the following hold\footnote{Note that unlike \cite{YoderLowChuang14}, where $\sqrt\lambda$ denotes the overlap, here we write the overlap as $\lambda$.}: 

\begin{lemma} \label{lem:FPSproperties}
	\begin{enumerate}[(i)]
		\item If $\lambda \geq  \lambda'$, then $|\!\bra{T,0}\mathrm{FPS}(\mathcal C, U, \lambda',\delta)\ket{0}^{\otimes (n'+1)}\!|^2 \geq 1-\delta^2$.
		\item If $\lambda\leq \lambda'$, then $|\!\bra{T,0}\mathrm{FPS}(\mathcal C, U, \lambda',\delta)\ket{0}^{\otimes( n'+1)}\!| \leq 2 \frac\lambda{\lambda'} \ln\frac2\delta $. 
	\end{enumerate}
\end{lemma}
\begin{proof} 
	Part (i) is the central result proven in \cite{YoderLowChuang14}. To prove part (ii), let ${t }:=\lceil \frac1{\lambda'} \ln\frac2\delta\rceil$ be the number of calls to $\mathcal C$. From \cite[Eq (1)]{YoderLowChuang14},  the success probability $P$ can be expressed in terms of (generalised) Chebyshev polynomials $\mathcal T_t(x)$ of the first kind,
	\begin{equation}
	 P:= |\!\bra{T,0}\mathrm{FPS}(\mathcal C, U, \lambda',\delta)\ket{0}^{\otimes( n'+1)}\!|^2
		= 1- \delta^2 \mathcal T_{t } \left( \mathcal T_{1/{t }} (1/\delta)\sqrt{1-\lambda^2}\right)^2.
	\end{equation}
	W.l.o.g. assume that ${t }\geq 5$, so ${t }({t }+1)\leq 2\left(\frac1{\lambda'}\log\frac2\delta\right)^2$.  It thus suffices to prove that $P\leq 2{t }({t }+1)\lambda^2$. Write $\lambda = \sin\theta$ and $\tau := \mathcal T_{1/{t }}(1/\delta) = \cosh({t }^{-1}\arccosh(1/\delta))\geq 1$. Note that $\mathcal T_{t }(\tau) = 1/\delta$. Then, using the mean value theorem on the function $\mathcal T_{t }(x)^2$, we obtain
	\begin{align}
		P &= 1- \frac{\mathcal T_{t }(\tau\cos\theta)^2}{\mathcal T_{t }(\tau)^2} \\
		&= \tau(1-\cos\theta) \frac{ 2\mathcal T_{t }(\xi) \mathcal T_{t }'(\xi)}{\mathcal T_{t }(\tau)^2} \\
		&= {t }\tau(1-\cos\theta) \frac{ 2\mathcal T_{t }(\xi) \mathcal U_{{t }-1}(\xi)}{\mathcal T_{t }(\tau)^2}
	\end{align}	 
	for some $\xi\in[\tau\cos\theta,\tau]$, where we used that $\mathcal T_{t }'(x) = {t }\mathcal U_{{t }-1}(x)$. 
	Note that since $\tau\geq 1$ and $\xi\leq \tau$, $|\mathcal T_{t }(\xi)|\leq \mathcal  T_{t }(\tau)$ and $|\mathcal U_{{t }-1}(\xi)|\leq \mathcal U_{{t }-1}(\tau)$. Indeed, both $\mathcal T_{t }(x)$ and $\mathcal U_{{t }-1}(x)$ attain their maximum moduli on $[-1,1]$ at $x=1$, and are monotonically increasing on $[1,\infty)$. 
	Finally, using the relations 
	\begin{equation}
		\mathcal U_{t }(\tau) = 
			\begin{cases}
				1+ 2\sum_{k=1}^{{t }/2} \mathcal T_{2k}(\tau) & \text{${t }$ even} \\
				 2\sum_{k=0}^{({t }-1)/2} \mathcal T_{2k+1}(\tau) & \text{${t }$ odd} 
			\end{cases}
	\end{equation}
	and $0<\mathcal T_{k'}(\tau) \leq \mathcal T_k(\tau)$ for $k'\leq k$ and $\tau\geq 1$, we obtain
	\begin{equation}
		2\tau \mathcal U_{{t }-1} (\tau) = \mathcal U_{t }(\tau) + \mathcal U_{{t }-2}(\tau) \leq 2\mathcal U_{t }(\tau) \leq 2 ({t }+1) \mathcal T_{t }(\tau).
	\end{equation}
	Hence,
	\begin{align}
	 P &\leq 2{t }({t }+1) (1-\cos\theta ) \leq2{t }({t }+1)\sin^2\theta = 2{t }({t }+1)\lambda^2,
	\end{align}
	as claimed. 
\end{proof}

\begin{proof}[Proof of Proposition~\ref{prop:minFinding}]
	The algorithm proceeds by successively attempting to find the binary digits $a_1,\ldots,a_n\in\{0,1\}$ of an integer such that $2^{n-1}a_1+\cdots+a_n\in[\tilde J, J]$. The algorithm runs in two stages. In the first stage of the algorithm, suppose that $a_1,\ldots, a_{k-1}$ have already been found. To obtain $a_k$, we use fixed point search \cite{YoderLowChuang14} to search for $\ket{0}^{\otimes q}\ket{a_1\ldots a_{k-1}0}$ on the first $q+k$ qubits, using at most $O(\frac1\zeta\log\frac 1{\delta'})$ repetitions of $\mathcal C_\Phi$. We repeat this search $K$ times. We choose $\delta'=\delta/(2n\log(n/\delta))$ and $K=\lceil\log(n/\delta)\rceil$. If it succeeds all $K$ times, we set $a_{k}=0$ and move on to the next digit $a_{k+1}$ (unless $k=n$, in which case the algorithm terminates). Otherwise, we search $K$ times for $\ket{0}^{\otimes q}\ket{a_1\ldots a_{k-1}1}$ on the first $q+k$ qubits.  If we succeed all $K$ times, we set $a_k=1$ and move on to the next digit $a_{k+1}$ (unless $k=n$, in which case the algorithm terminates). 
	If we fail at least once, we say that the result is 'inconclusive', and we abort the first stage of the algorithm and move on. In the second stage of the algorithm, we successively search for $\ket{0}^{\otimes q}\ket{a_1\ldots a_{k-1}}, \ket{0}^{\otimes q}\ket{a_1\ldots a_{k-2}}$, etc., where each search is repeated $K$ times, using $O(\frac1\zeta\log \frac1{\delta'})$ calls to $\mathcal C_\Phi$, until we succeed. Once successful, we measure the remaining ancillas and the algorithm terminates. 
	
	We now show that this algorithm produces the required results. Let $J=2^{n-1} b_1+2^{n-2}b_2+\cdots+b_n$ and $\tilde J = 2^{n-1}c_1+2^{n-2}c_2+\cdots +c_n$ be the binary representations of $J,\tilde J$ respectively. Let $i_0$ be the maximum integer such that $b_l=c_l$ for all $l=1,\ldots,i_0$. Since $\tilde J < J$, it follows that $i_0<n$ and $b_{i_0+1} = 1, c_{i_0+1}=0$. 
	Moreover, we say that binary digits $a_1,\ldots,a_k$ are 'consistent' if they are the dominating $k$ binary digits of at least one integer in $[\tilde J,J]$. 
	
	We first show that the first stage of the algorithm finds at least $i_0$ consistent digits with probability at least $1-3\delta$. For a given $k$, the probability of finding the wrong $a_k$ given that $a_1,\ldots,a_{k-1}$ are consistent, is at most $\delta/n$. Indeed, finding too large a value means that $a_l=b_l$ for $l=1,\ldots,k-1$, $b_k=0$ and finding $a_{k}=1$. The probability of not finding $\ket{0}^{\otimes q}\ket{b_1\ldots b_k}$ $K$ times is $\delta'^K<\delta/n$ \cite{YoderLowChuang14}. 
	On the other hand, from Lemma~\ref{lem:FPSproperties}(ii), the probability of finding too small a value for a single trial of the search is upper bounded by 
	\begin{align}
		\frac{2\ln^2\frac2\delta}{\zeta^2}\sum_{j=0}^{\tilde J} \|\ket{\Phi_j}\|^2 &< \frac{\delta}{2n\log\frac n\delta} \label{eq:binSearchSumBound} \\
		&< \frac12.
	\end{align}
	 Thus, the probability of finding too small a value is at most $2^{-K}<\delta/n$.
	Moreover, for a given $k<i_0$, the probability of finding an inconclusive result is at most $\delta' K<\delta/n$. Thus, the probability of not finding at least $i_0$ consistent digits in the first stage of the algorithm is at most
	\begin{equation} \label{eq:binSearchProbFailStage1}
		n(\delta'^K + 2^{-K} + \delta'K) < 3\delta.
	\end{equation}
	
	Suppose now that we found $k\geq i_0$ consistent digits and then found an inconclusive result. Let $i_1\leq k$ be the largest integer such that $a_l=b_l$ for $l=1,\ldots,i_1$. Clearly, $i_1\geq i_0$. If $i_1<k$, then $a_{i_1+1}=0$ and $b_{i_1+1}=1$, since $a_1,\ldots,a_k$ are consistent. Thus, for all $l>i_1$, if the algorithm successfully finds $\ket{0}^{\otimes q}\ket{a_1\ldots a_l}$, then upon measuring the other ancillas, we obtain a value $j\leq J$. For a single trial of the search, the probability of finding a value $j<\tilde J$ is at most \eqref{eq:binSearchSumBound}. Thus, for a given $l$, the total probability of finding a value $j<\tilde J$ is at most $K$ times \eqref{eq:binSearchSumBound}, i.e. at most $\delta/n$. 

	The last possibility for the algorithm to fail is if the algorithm fails to find $\ket{a_1\ldots a_{i_1}}$ in the second stage. Let $p_{i_1}$ be the probability of a single run of the search finding $\ket{0}^{\otimes q}\ket{a_1\ldots a_{i_1}}$. Then, conditional on reaching this stage of the algorithm, the probability of failing is $(1-p_{i_1})^K$. Notice however that this event is conditional on first successfully finding $\ket{0}^{\otimes q}\ket{a_1\ldots a_{i_1}} $ $K$ times in the first stage of the algorithm. The probability of this event is $p_{i_1}^K$. Thus, the overall probability of the algorithm failing in this way is $(p_{i_1}(1-p_{i_1}))^K \leq 4^{-K}$. Thus, the probability of the algorithm failing in the second stage is at most
	\begin{equation}\label{eq:binSearchProbFailStage2}
		\delta + \sum_{i_1=1}^{n} (p_{i_1}(1-p_{i_1}))^K \leq \delta + n4^{-K} \leq 2\delta,
	\end{equation}
	where the first $\delta$ comes from the upper bound on the probability of finding $j<\tilde J$ in the second stage. 
	Thus, the total probability of the algorithm failing is at most \eqref{eq:binSearchProbFailStage1} plus \eqref{eq:binSearchProbFailStage2}, as claimed. The number of calls to $\mathcal C_\Phi$ and additional elementary gates is 
	\begin{equation}
		O\left(nK\frac1\zeta\log\frac1{\delta'}\right) = O\left(\frac{n}{\zeta} \log^2\frac n\delta\right),
	\end{equation}
	as claimed. 
	\end{proof}

Notice that if only the first $n'< n$ dominant binary digits of an integer in $[\tilde J,J]$ is needed, $n$ can be replaced with $n'$ in \eqref{eq:AppALem1Sum} and \eqref{eq:AppALem1Gates}. Moreover, it is clear that the $\zeta^2$ dependence in \eqref{eq:AppALem1Sum} is optimal. 

\subsection{Proof of Theorems~\ref{thm:GSEunknown}--\ref{thm:estimation}} \label{subsec:proofsUnknown}

With $M$ and $\tau$ defined as in Section~\ref{sec:known}, let $L = \Theta(\sqrt{M})$ be a power of $2$ and define $E_j := j / L$. Clearly, $E_{j+1}-E_j \in O(1/\sqrt{M})$. 
Define $\delta_j := \lambda_0 - E_j$ and $H_j := \tilde{H} - (E_j - \tau) \id$. Let $J \in \{0, \ldots, L-1\}$ be the largest integer such that $E_J \leq \lambda_0$. Clearly, $\delta_J \in O(1/\sqrt{M})$. Using the algorithm from Section~\ref{sec:known}, we have unitaries $V_j$ such that
\begin{equation}\label{eq:Vj}
V_j \ket{0}^{\otimes(b+q)}\ket\phi = \ket0^{\otimes (b+q)}\mathcal G_j\ket\phi +  \ket{R},
\end{equation}
where  
\begin{equation}\label{eq:Gj}
\mathcal G_j=\frac1\alpha\sum_{k=-m_0}^{m_0}\alpha_k e^{-2iH_j k},
\end{equation}
and $(\ketbra00^{\otimes (b+q)} \otimes \id )\ket{R}=0$. 
The gate cost of a single $V_j$ is given by 
\begin{equation}\label{eq:costVj}
	\tilde O\left(\frac{\Lambda}{\Delta}\right).
\end{equation}

We now introduce an additional ancilla system $\mathcal L$ on $l:=\lceil\log_2 L\rceil$ qubits. 
Let $V$ be the $V_j$'s controlled on $\mathcal L$, i.e. $V = \sum_{j=0}^{L-1} \ketbra jj_{\mathcal L} \otimes V_j$. Then, 
\begin{equation}\label{eq:Vphi}
V\ket{+}^{\otimes l}_{\mathcal L} \ket{0}^{\otimes(b+q)}\ket\phi 
= \frac1{\sqrt L}\sum_{j=0}^{L-1}\ket j_{\mathcal L} \ket{0}^{\otimes (b+q)}\mathcal G_j \ket\phi + \ket R,
\end{equation}
where $(\id_{\mathcal L} \otimes \ketbra00^{\otimes (b+q)}\otimes\id)\ket R =0$. 
From \eqref{eq:cosFullApprox} and \eqref{eq:truncFourier}, 
\begin{equation} \label{eq:normGj}
	\left\| \frac{\mathcal G_j\ket\phi}{\| \mathcal G_j\ket\phi\|} -\ket{\lambda_0} \right\| = O(\epsilon),
\end{equation}
whenever $ E_j \leq  \lambda_0$, and
$
\| \mathcal G_j\ket\phi\| = \Theta\left(|\phi_0|\cos(\tau + \delta_j)^M\right)
$.
In particular, $\|\mathcal G_J\ket\phi\| = \Omega(|\phi_0|)$. 

We now show that we can use the minimum label finding algorithm  to project \eqref{eq:Vphi} onto the state $\ket j_{\mathcal L} \mathcal G_j \ket\phi / \|\mathcal G_j\ket\phi\|$ for some $j\in[\tilde J, J]$ and suitable $\tilde J$. Indeed, for any integer $\tilde J$,
\begin{align}
	\sum_{j=0}^{\tilde J} \|\mathcal G_j\ket\phi\|^2 &= O\left( \sum_{j=0}^{\tilde J}|\phi_0|^2 \cos(\tau + \delta_j)^{2M}\right) \\
	&= O\left( |\phi_0|^2\sum_{j=0}^{\tilde J} e^{-2(\tau + \delta_j)^2M/3}\right) \\
	&= O\left( |\phi_0|^2\sum_{j=0}^{\tilde J} e^{-4\delta_j\tau M/3}\right)\\
	&= O\left( |\phi_0|^2\sum_{j=0}^{\tilde J} e^{-4(\lambda_0-j/L)\tau M/3}\right)\\
	&= O\left( |\phi_0|^2 e^{-4(\lambda_0-\tilde J/L)\tau M/3}\right)\\
	&= O\left( |\phi_0|^2 e^{-4\delta_{\tilde J}\tau M/3}\right)\\
	&= O\left( |\phi_0|^2 e^{-\Theta(\delta_{\tilde{J}}\sqrt{M})}\right), 
\end{align}
where the last step follows from %
\eqref{eq:choicetau}. 
Thus, since $\chi$ is a known lower bound on $|\phi_0|$ satisfying $\log(|\phi_0|/\chi) = O(\log N)$, we can ensure that
\begin{equation}
	\sum_{j=0}^{\tilde J} \|\mathcal G_j\ket\phi\|^2 =O\left(\frac{\chi^2}{\log M}\right),
\end{equation}
provided that
\begin{equation} \label{eq:deltaTildeJ}
	\delta_{\tilde J}= \Theta\left(\frac{1}{\sqrt{M}} \log \frac{|\phi_0|^2 \log M}{\chi^2}\right) = \tilde \Theta(\Delta).
\end{equation}
Thus, by  Proposition~\ref{prop:minFinding}, we can with high probability prepare the state $\ket j_{\mathcal L} \mathcal G_j \ket\phi / \|\mathcal G_j\ket\phi\|$ for some $j\in[\tilde J, J]$ using $\tilde O(\sqrt L/\chi)$ repetitions of $V$ and $\mathcal C_\phi$. The value in the $\mathcal L$ register gives an estimate of the ground energy to a precision of \eqref{eq:deltaTildeJ}. From \eqref{eq:normGj}, the second register then contains the desired approximation of the ground state.

Notice that $V$ can be implemented with essentially the same cost as a single $V_j$. 
Indeed, the only explicit dependence of $V_j$ on $j$ is in the call to Hamiltonian simulation, $e^{\pm 2iH_j}= e^{\pm 2i(\tau-E_j)} e^{\pm 2i\tilde H}$. Thus, given access to $e^{\pm 2i\tilde H}$, it is easy to implement 
\begin{equation}
	\sum_{j=0}^{L-1}\ketbra jj \otimes e^{\pm 2i H_j} = \sum_{j=0}^{L-1}\ketbra jj \otimes e^{ \pm 2i (\tau-E_j)} e^{ \pm 2i \tilde H}
\end{equation}
with a single call to $e^{\pm 2i\tilde H}$ and $O(\log L)=\tilde O(1)$ additional elementary gates implementing the phases $e^{\pm2 i(\tau-E_j)}$. 
Hence, the overall gate cost of this algorithm is 
\begin{equation}
	\tilde O\left(\frac{\sqrt L}\chi \left(\frac\Lambda{\Delta} + \Phi\right)\right) = \tilde O\left(\frac{\Lambda}{\chi\Delta^{3/2}} + \frac\Phi{\chi\sqrt\Delta}\right),
\end{equation}
 as claimed. This proves Theorem~\ref{thm:GSEunknown}.


Alternatively, one can combine this algorithm with the prior use of phase estimation. This approach improves the scaling with $\Delta$ at the cost of a worse scaling in $\chi$, and is hence useful if $\Delta$ is very small. 

First we use the method from Appendix~\ref{app:PEAenergy} to obtain a ``crude'' estimate of the ground energy, to a precision of $\xi=O(\Delta^\kappa)$ for some $\kappa\in[0,1]$ to be chosen later. The runtime of this is 
\begin{equation}\label{eq:unknownPEACost}
 \tilde O\left(\frac{\Lambda}{\chi^3\Delta^\kappa} + \frac{\Phi}\chi\right)
\end{equation}
gates (see Appendix~\ref{app:PEAenergy}). 
This provides us with an interval $[a,b]\ni \lambda_0$ with $b-a = O(\Delta^{\kappa})$. Let $ E_j' = a+(b-a)j/L'$, where $L'=\Theta(M\Delta^{\kappa})$. 
Writing $H_j' = \tilde H - (E_j'-\tau)\id$ and 
\begin{equation}
	\mathcal G_j' = \frac1\alpha \sum_{k=0}^{m_0} \alpha_k e^{-2iH_j'k},
\end{equation}
we run the previous algorithm but with
\begin{equation} \label{eq:Vprimephi}
 V'\ket+^{\otimes l'}_{\mathcal L'}\ket0^{\otimes(b+q)} \ket \phi = \frac1{\sqrt {L'}} \sum_{j=0}^{L'-1}\ket j_{\mathcal L'}\ket 0^{\otimes (b+q)}\mathcal G_j'\ket\phi + \ket {R}
\end{equation}
instead of \eqref{eq:Vphi}, and an ancilla system $\mathcal L'$ on $l'=\lceil \log_2 L'\rceil$ qubits instead of $\mathcal L$. 
The algorithm now requires $O(\sqrt{L'}/\chi)$ repetitions of $V'$. As before, the cost of $V'$ is given by \eqref{eq:costVj}. Hence, the number of gates for projecting \eqref{eq:Vprimephi} onto the target state is
\begin{equation}\label{eq:unknownPowerCost}
 \tilde O\left(\frac{\sqrt{L'}}{\chi}\left( \frac\Lambda{\Delta} + \Phi\right)\right)
 = \tilde O\left( \frac{\Lambda}{\chi \Delta^{(3-\kappa)/2}} + \frac{\Phi}{\chi\Delta ^{(1-\kappa)/2}} \right).
\end{equation}
The total number of gates for the algorithm is thus \eqref{eq:unknownPowerCost} plus the gate cost \eqref{eq:unknownPEACost} from the prior use of phase estimation, i.e.,
\begin{equation}
 \tilde O\left( \frac{\Lambda}{\chi^3\Delta^\kappa} + \frac{\Lambda}{\chi \Delta^{(3-\kappa)/2}} + \frac{\Phi}{\chi\Delta ^{(1-\kappa)/2}} \right).
\end{equation}
This proves Theorem~\ref{thm:combined}.

It is easy to see that both algorithms can be used to estimate the ground energy to an arbitrarily small precision $\xi$, provided that $\xi$ is less than \eqref{eq:deltaTildeJ}. Indeed, the entire algorithm has $\Delta$ as an input parameter, which is only required to be a lower bound  on the spectral gap rather than the spectral gap itself. Hence, if we run the algorithm with $\Delta'\ll \Delta$, we obtain an estimate of the ground state to a precision of $\xi=\tilde O(\Delta')$. The algorithm can thus be used to estimate the ground state to an arbitrary precision $\xi < \Delta$ in a gate complexity of 
\begin{equation}
\tilde	O\left( \frac{\Lambda}{\chi\xi^{3/2}} + \frac\Phi{\chi\sqrt\xi}\right),
\end{equation}
for the first algorithm, and 
	\begin{equation} 
 	\tilde O\left( \frac{\Lambda}{\chi^3\xi^\kappa} + \frac{\Lambda}{\chi \xi^{(3-\kappa)/2}} + \frac{\Phi}{\chi\xi^{(1-\kappa)/2}} \right)
	\end{equation}
for the combined approach. This proves Theorem~\ref{thm:estimation}. \qed 

Note that choosing $\kappa=1$ in the combined algorithm gives the optimal scaling of $\sim 1/\Delta$ with $\Delta$. However, it is obviously also possible to choose different values of $\kappa$ that also take the other parameters into account.

\section{Conclusion} \label{sec:conclusion}

In this paper, we presented quantum algorithms to prepare the ground state of a Hamiltonian faster and with fewer qubits than with previous methods, both in the case of known and unknown ground energy. In the latter scenario, the algorithm also provides a high precision estimate of the ground energy in a complexity faster than with phase estimation and amplitude amplification. 

Perhaps surprisingly, the straightforward use of phase estimation and amplitude amplification has a significantly worse scaling in the overlap of the trial state with the ground state than what one would naively expect.  In Appendix~\ref{app:PEApreparation}, we show that straightforward phase estimation requires
\begin{equation} \label{eq:PEAgatesPrep}
\tilde O\left(\frac  {\Lambda}{|\phi_0|^2 \Delta\epsilon}+\frac{\Phi}{|\phi_0|}\right)
\end{equation}
gates to prepare the ground state, provided the ground energy is known to a precision of $O(|\phi_0|\epsilon\Delta)$ beforehand. Notice in particular the $ 1/|\phi_0|^2$ scaling, even after using amplitude amplification. The scaling becomes even worse if the ground energy is not known beforehand (see Table~\ref{tab:CostUnknown}a). 

Previous improvements by Poulin and Wocjan \cite{PoulinWocjan09} to the phase estimation approach quadratically improved the dependence on $|\phi_0|$. Moreover, we prove in Appendix~\ref{app:PoulinWocjan} that for suitable choices of parameters, this algorithm can prepare the ground state 
in the same runtime as our algorithm. 
The algorithms in this paper however use significantly fewer qubits, which make them more attractive for early applications of quantum computing.

In the case of known ground energy, no non-trivial prior knowledge of the value of $|\phi_0|$ is required for any algorithm discussed in this paper (although for the filtering method, the number of qubits required becomes worse if no non-trivial lower bound is known). 
However, in case of unknown ground energy, the runtime dependences on $|\phi_0|$ are replaced by dependences on the known lower bound $\chi$ of $|\phi_0|$. There appears to be a systematic reason for this which seems difficult to overcome: broadly speaking, all methods discussed in this paper produce a state of the form $\sum_j\ket{E_j}\ket{\Phi_j}$, where $\ket{\Phi_j}$ is approximately proportional to the ground state and $\|\ket{\Phi_j}\| \approx |\phi_0|$ if $E_j\approx\lambda_0$, and $\|\ket{\Phi_j}\|\ll |\phi_0|$ if $E_j\ll \lambda_0$. Any search for the ``correct'' value of $j$ needs to be able to distinguish the following two cases: (i)  $\lambda_0=\lambda_0^{(1)}$ and $\phi_0=\phi_0^{(1)}$, and (ii)  $\lambda_0=\lambda_0^{(2)}$ and $\phi_0=\phi_0^{(2)}$, where $\lambda_0^{(1)} \ll \lambda_0^{(2)}$ and $|\phi_0^{(1)}| \ll |\phi_0^{(2)}|$. Any use of amplitude amplification to distinguish these cases would require many, i.e. $O\left(1/|\phi_0^{(1)}|\right)$ repetitions to ``see'' if (i) is the case, which means that if in fact (ii) is the case, significantly more repetitions than $O\left(1/|\phi_0^{(2)}|\right)$ have been performed. On the other hand, if only  $O\left(1/|\phi_0^{(2)}|\right)$  repetitions are performed, one would accidentally find a much larger value (e.g. $\lambda_0^{(2)}$) than the true ground energy $\lambda_0^{(1)}$ (and hence also prepare the wrong state) if in fact (i) is the case. We currently do not see a way to get around this problem. 

For simplicity of notation, we assumed throughout this paper that the ground energy of $\tilde H$ is non-degenerate. The algorithm however trivially generalises to degenerate ground spaces. Indeed, the only thing that changes in this case is that the resulting state $\ket{\lambda_0'}$ is $\epsilon$-close to $\ket{\lambda_0} = P\ket\phi/\|P\ket\phi\|$, where $P$ is the projector onto the ground space of $\tilde H$. The algorithm also generalises to the case of only approximately degenerate ground spaces, i.e. when the spectrum of $\tilde H$ is contained in $[\lambda_0,\lambda_0+\varepsilon] \cup [\lambda_0+\Delta,1]$ with $\varepsilon \ll \Delta$.

Finally, given the close relation of this work to \cite{ChildsKothariSomma15}, it is also natural to expect that the alternative approach of \cite{ChildsKothariSomma15} using Chebyshev polynomials and quantum walks can be used to get a different algorithm with the same runtime scaling. 
This is indeed possible, as we show in Appendix~\ref{app:Chebyshev}. These algorithms however are restricted to the case where $\tilde H$ is a sparse Hamiltonian with quantum oracle access. 
 By contrast, the advantage of the approach based on Hamiltonian simulation is that it can be applied outside the framework of given oracle access to $\tilde H$, as efficient Hamiltonian simulation algorithms exist for other models of accessing the Hamiltonian \cite{LowChuang16,BCCKS15}.

\acknowledgments
 YG thanks H.~Buhrman, V.~Dunjko, A.~Gily\'en, D.~Gosset, R.~Kothari, G. Low, A.~Moln\'ar, A.~Montanaro, and N.~Schuch for helpful discussions. 
This work was supported by the ERC grant QUENOCOBA, ERC-2016-ADG, grant no. 742102. 
JT acknowledges funding from the EU's Horizon 2020 research and innovation programme under the Marie
Sk{\l}odowska-Curie grant agreement No 748549.

%

\appendix

\section{Finding the ground energy using phase estimation} \label{app:PEAenergy}

We now demonstrate how to combine phase estimation and the minimum label finding algorithm to find the unknown ground energy of a quantum Hamiltonian to a precision of $\xi$ in a gate complexity of \begin{equation} \label{eq:PEAEstGates}
 \tilde O\left( \frac\Lambda{\chi^3\xi} +\frac\Phi\chi \right). 
\end{equation}%
Notice that straightforward amplitude amplification cannot be used to amplify the ground energy since the latter is not known.

Let $U$ be an $N\times N$ unitary with eigenvectors $\ket{\lambda_i}$ and eigenvalues $e^{2\pi i \lambda_i}$ with $\lambda_i\in[0,1)$. We are interested in finding a good approximation  of the minimum value $\lambda_0$, say to $n = \lceil\log_2 1/\xi\rceil$ binary digits of precision. %

Let $\ket\phi = \sum_i\phi_i \ket{\lambda_i}$ be an arbitrary trial state. Suppose that we are given a circuit $\mathcal C_\phi$ 
which prepares $\ket\phi$ using $\Phi$ elementary gates.
Recall that the phase estimation algorithm  \cite{Kitaev95} takes $\ket\phi  \ket 0^{\otimes k}$ for some $k>n$ to
\begin{equation} \label{eq:PEAresult} 
	\ket\Phi:=\sum_i \phi_i\ket{\lambda_i} \ket{\tilde\varphi_i}
\end{equation}
using controlled $U,U^2,\ldots U^{2^{k-1}}$ and an inverse quantum Fourier transform on $k$ qubits,
where $\ket{\tilde\varphi_i}$ can be shown to have a large overlap with the computational basis state that encodes the first $n$ binary digits of $\lambda_i$. 

Let 
\begin{equation}
	\ket{\tilde\varphi_i} = \sum_{x=0}^{2^k-1} \gamma_{ix}\ket x
\end{equation}
be its expansion in the computational basis. 
Then, 
\begin{equation}
	\ket\Phi = \sum_{x=0}^{2^k-1} \ket{\Phi_x}\ket x,
\end{equation}
where 
\begin{equation}
\ket{\Phi_x} = \sum_{i} \phi_i \gamma_{ix} \ket{\lambda_i}
\end{equation}
are non-normalised states.

Naively, one would expect that simply measuring the ancillas would give a good approximation $x\approx 2^k\lambda_0$ of the ground energy with probability $\approx |\phi_0|^2$, so that $O(1/|\phi_0|^2)$ repetitions would be needed to find the ground energy. One would also expect that this could be quadratically reduced to $O(1/|\phi_0|)$ using suitable amplitude amplification techniques. However, we show below that the overall gate cost is much higher if the ground energy is not known beforehand. Indeed, due to the finite precision and error in the phase estimation algorithm, $\|\ket{\Phi_x}\|^2$ could in principle be large even when $x\ll 2^k\lambda_0$. This means that the algorithm could fail by  accidentally finding a value that is much smaller than the true ground energy. We show below that in order to guarantee that this does not happen, one needs to choose $2^k = \tilde O(2^n /{|\phi_0|^2})$. Since the runtime of phase estimation scales linearly with $2^k$, and only a lower bound $\chi$ on $|\phi_0|$ is assumed to be known, this would lead to an overall runtime that scales as $\sim 1/{\chi^3}$. We also demonstrate at the end of this  section that this dependence is essentially tight.

\begin{proposition}[Ground energy estimation with phase estimation] \label{prop:PEAEst}
	Phase estimation can be used to find the value of $\lambda_0$ to an additive precision of $\xi$ in a gate complexity of 
	\eqref{eq:PEAEstGates}, and 	using
	$
	 O(\log N+\log 1/{\xi})
	$
	qubits. 
\end{proposition}
\begin{proof}
Recall first the well-known  \cite{NielsenChuang} relations
\begin{equation}
	\gamma_{ix} = \frac1{2^k} \left( \frac{1-e^{2\pi i (2^k\lambda_i - x)}}{1-e^{2\pi i (\lambda_i - x/2^k)}} \right)
\end{equation}
and
\begin{equation} \label{eq:upperboundgamma}
	|\gamma_{ix}| \leq \frac1{2 |2^k\lambda_i - x|}.
\end{equation}
Let $D:=2^{k-n}$. 
Then, using  \eqref{eq:upperboundgamma}, we have
\begin{align}
  \sum_{x < 2^k\lambda_0 -D} \|\ket{\Phi_x}\|^2 
  &=  \sum_{\substack{i\\  x < 2^k\lambda_0 - D}} |\phi_i|^2|\gamma_{ix}|^2 \\
  &\leq \sum_{\substack{i\\  x < 2^k\lambda_0 - D}} |\phi_i|^2\left|\frac1{2 |2^k\lambda_i - x|}\right|^2 \\
  &\leq \sum_{\substack{i\\  x < 2^k\lambda_0 - D}} |\phi_i|^2\left|\frac1{2 |2^k\lambda_0 - x|}\right|^2  \\
  &\leq \sum_{ x < 2^k\lambda_0 - D} \left|\frac1{2 |2^k\lambda_0 - x|}\right|^2 \\
  &< \sum_{y>D} \frac{1}{y^2} < \frac 1D.
\end{align}

Moreover, when $x$ satisfies $|\lambda_i-x/2^k| < 1/2^{k+1}$, 
\begin{equation} \label{eq:lowerboundgamma}
	|\gamma_{ix} | \geq \frac2{\pi},
\end{equation}
using that 
\begin{equation}\label{eq:thetabound}
	|\theta|\geq |1-e^{i\theta}| \geq 2|\theta|/\pi
\end{equation}
for $\theta\in[-\pi,\pi]$. Thus,
\begin{equation}
	\left\|\ket{\Phi_{\lfloor 2^k\lambda_0\rceil}} \right\| \geq \frac2\pi |\phi_0|.
\end{equation}

Let $\chi$ be a known lower bound on $|\phi_0|$. Using Proposition~\ref{prop:minFinding}, we can thus find an integer $x\in [ \lfloor 2^k\lambda_0\rceil - D, \lfloor 2^k\lambda_0\rceil ]$ with good probability, provided that $1/D=\tilde O(\chi^2)$. 
The number of calls to phase estimation is $\tilde O(1/\chi)$. The number of gates required for each run of phase estimation is $\tilde O(2^k\Lambda)$. Thus, 
taking $U=e^{-2\pi i\tilde H}$ and $\xi=2^{-n}$, we arrive at  
a total gate count of \eqref{eq:PEAEstGates}. Moreover, the number of qubits required is $O(\log N + k) = O(\log N  + \log 1/\xi)$. 
\end{proof}

It is  not difficult to show 
that for suitable $ \tilde H$ and $\ket\phi$, the dependence on $\chi$ of this algorithm is optimal. 
Indeed, suppose 
that $\tilde H = \frac1{2^k}(H'+c\id)$, where $c\in(0,1/2)$ is a constant and $H'$ is any Hamiltonian with integer spectrum in $\{2^{k-2},\ldots,2^{k-1}-1\}$. 
Then from \eqref{eq:thetabound},
\begin{equation}
	|\gamma_{ix}| \geq \frac{2c}{\pi|2^k\lambda_i-x|}.
\end{equation}
Suppose  that our trial state  $\ket\phi$ is such that $|\phi_1|^2\geq 1/2$.
Then, 
\begin{align}
 \sum_{x < 2^k\lambda_0 -D} \|\ket{\Phi_x}\|^2   
  &=  \sum_{\substack{i\\  x < 2^k\lambda_0 - D}} |\phi_i|^2|\gamma_{ix}|^2 \\
  &\geq \frac c\pi \sum_{  x < 2^k\lambda_0 - D}\frac1{|2^k\lambda_1 - x|^2} \\
  &= \Omega\left(\int_{D+2^k\Delta}^{2^k\lambda_1} \frac1{x^2}\mathrm{d}x \right) \\
  &=\Omega\left( \frac1{2^k}\left(\frac1{\xi+\Delta}- \frac1{\lambda_1} \right)\right).
\end{align}
Suppose  now that $\Delta = O(\xi)$ and $\xi = o(1)$. Then
\begin{equation}
	\sum_{x < 2^k\lambda_0 -D} \|\ket{\Phi_x}\|^2  = \Omega\left(\frac1{2^k\xi}\right),
\end{equation}
which is much larger than $\chi^2$ unless $2^k = \Omega\left(\frac1{\chi^2\xi}\right)$, and the claim follows.

\section{Preparing the ground states using phase estimation} \label{app:PEApreparation}

Naively, one would  expect that upon successfully projecting the ancillas of \eqref{eq:PEAresult}  into the first $n$ (or even $k$) digits of $\lambda_0$, the residual state should be a good approximation of $\ket{\lambda_0}$, provided that $\xi =\tilde O( \Delta)$. Unfortunately, this is not true, because the ``imperfections'' in $\ket{\tilde\varphi_i}$ build up to a non-negligble error in the residual state. This phenomenon was also observed in \cite{PoulinWocjan09}. 
Below, we analyse these errors and show a non-negligible bound on the minimum precision to which the ground energy needs to be known beforehand.

\begin{proposition}[Ground state preparation with phase estimation for known ground energy]
	Suppose that the value of $\lambda_0$ is known to a precision of $O(|\phi_0|\epsilon\Delta)$. Then phase estimation can be used to prepare a state $\epsilon$-close to the ground state in a gate complexity of
	\begin{equation}\label{eq:PEAknownPrepGates}
	 O \left(\frac\Lambda{|\phi_0|^2\Delta\epsilon}+\frac\Phi{|\phi_0|}\right),
	\end{equation}	 
	and using $O(\log N + \log 1/\epsilon + \log 1/\Delta)$ qubits. 
\end{proposition}
\begin{proof}
Suppose that $z:= \lfloor 2^k\lambda_0\rceil$, i.e. $z$ encodes the first $k$ binary digits of $\lambda_0$ ($k$ will be specified below) 
and we apply phase estimation on $k$ qubits. Then, post-selecting the ancillas of \eqref{eq:PEAresult} to be in $\ket{z}$, the residual state is
\begin{equation}
	\ket{\lambda} := \frac{\sum_i \phi_i \gamma_{iz} \ket{\lambda_i} }{\left\|\sum_i \phi_i \gamma_{iz} \ket{\lambda_i}\right\|}.
\end{equation}
Thus, the error in the residual state is
\begin{equation}
\|\ket\lambda - \ket{\lambda_0}\| = \Theta\left( \frac{\|\sum_{i\neq 0} \phi_i \gamma_{iz} \ket{\lambda_i}\|}{\left\|\sum_i \phi_i \gamma_{iz} \ket{\lambda_i}\right\|}\right).
\end{equation}
For $i\neq 0$, \eqref{eq:upperboundgamma} implies that 
\begin{equation} \label{eq:gammagap}
	|\gamma_{iz} | \leq \frac{1}{2^{k+1}\Delta}. 
\end{equation}
Moreover, \eqref{eq:lowerboundgamma} implies $|\gamma_{0z}| \geq 2/\pi$. 
Hence,
\begin{align}
\left(\frac{\|\sum_{i\neq 0} \phi_i \gamma_{iz} \ket{\lambda_i}\|}{\left\|\sum_i \phi_i \gamma_{iz} \ket{\lambda_i}\right\|}\right)^2 &= \frac{\sum_{i\neq 0 } |\phi_i|^2|\gamma_{iz}|^2}{\sum_i |\phi_i|^2|\gamma_{iz}|^2} \\
&\leq \left(\frac1{\pi|\phi_0| 2^{k}\Delta}\right)^2.
\end{align}
Thus, it is sufficient if $k$ satisfies 
\begin{equation}\label{eq:2kboundPEAPrep}
	2^k= O\left(\frac{1}{|\phi_0|\epsilon\Delta}\right)
\end{equation}
 for an error $\|\ket\lambda - \ket{\lambda_0}\|=O(\epsilon)$. Notice in particular the scaling  with $|\phi_0|$. 

Suppose that the value of $z$, i.e. the ground energy to a precision of $k$ binary digits, is known, where $k$ is known to satisfy \eqref{eq:2kboundPEAPrep}. 
The cost of a single run of phase estimation is $\tilde O\left(2^k\Lambda\right)$, up to polylogarithmic factors. Fixed point search requires $O(1/|\phi_0|)$ applications of phase estimation and $\mathcal C_\phi$. We thus arrive at a total gate count of  
\begin{equation}
\tilde O\left(\frac{2^k\Lambda+\Phi}{|\phi_0|} \right) = \tilde O \left(\frac\Lambda{|\phi_0|^2\Delta\epsilon}+\frac\Phi{|\phi_0|}\right)
\end{equation}
as claimed. \end{proof} 

Suppose next that the ground energy is not known beforehand. If phase estimation and minimum label finding (Appendix~\ref{app:PEAenergy}) is first used to find the ground energy, we require a precision of $O(|\phi_0|\epsilon\Delta)$. Since $|\phi_0|$ is not assumed to be known, we need to run phase estimation to find the energy to a precision of $\xi=O(\chi\epsilon\Delta)$. Thus, from 
Proposition~\ref{prop:PEAEst}, the number of gates to first find the ground energy to the required precision takes 
\begin{equation}
 \tilde O \left(\frac\Lambda{\chi^4\Delta\epsilon}+\frac\Phi{\chi}\right)
\end{equation}
gates. 

\begin{corollary}[Ground state preparation with phase estimation for unknown ground energy]
 If the ground energy is not known beforehand, phase estimation can be used to prepare a state $\epsilon$-close to the ground state in a gate complexity of 
 \begin{equation}
 \tilde O \left(\frac\Lambda{\chi^4\Delta\epsilon}+\frac\Phi{\chi}\right),
\end{equation}
and using $O(\log N + \log 1/\epsilon + \log 1/\Delta)$ qubits. 
\end{corollary}
 
We now argue that the dependence of $2^k$ on $|\phi_0|$ in \eqref{eq:2kboundPEAPrep}, and thus the quadratic dependence on $|\phi_0|$ in \eqref{eq:PEAknownPrepGates},  is essentially tight for this algorithm. Suppose that
\begin{equation}
\frac{\sum_{i\neq 0 } |\phi_i|^2|\gamma_{iz}|^2}{\sum_i |\phi_i|^2|\gamma_{iz}|^2} \leq \epsilon^2. 
\end{equation}
Then,
\begin{align}
	\epsilon^2 &\geq \frac{\sum_{i\neq 0 } |\phi_i|^2|\gamma_{iz}|^2}{\sum_i |\phi_i|^2|\gamma_{iz}|^2} \\
	&= \frac{\sum_{i\neq 0 } |\phi_i|^2|\gamma_{iz}|^2}{|\phi_0|^2|\gamma_{0z}|^2+\sum_{i\neq 0} |\phi_i|^2|\gamma_{iz}|^2}\\
	&\geq (1-\epsilon^2)\frac{\sum_{i\neq 0 } |\phi_i|^2|\gamma_{iz}|^2}{|\phi_0|^2|\gamma_{0z}|^2}\\
	&\geq (1-\epsilon^2)\frac1{|\phi_0|^2} \sum_{i\neq 0 } |\phi_i|^2|\gamma_{iz}|^2.
\end{align}
Suppose that for all $i\neq 0$, 
\begin{equation}
	r(2^k\lambda_i -z) \geq c,
\end{equation}
 where $r(x) := |x-\lfloor x\rceil|\in [0,1/2]$, for some constant $c>0$. Notice that this can be achieved for any Hamiltonian of the form $H= \frac1{2^k} (\tilde H + c\id)$, where $\tilde H$ is any Hamiltonian with integer spectrum. Then, using Eq.~\eqref{eq:thetabound},
\begin{equation} \label{eq:gammalowerbound2}
	|\gamma_{iz}| = \frac1{2^k}\left(\frac{1-e^{2\pi i r(2^k\lambda_i-z)}}{1-e^{2\pi i (\lambda_i-z/2^k)}}\right)
	\geq \frac c{2^{k-1}}. 
\end{equation}
Thus,
\begin{equation}
 \epsilon^2 \geq (1-\epsilon^2)\frac{c^2(1-|\phi_0|^2)}{4^{k-1}|\phi_0|^2}.
\end{equation}
Therefore, $2^{k} = \Omega\left(\frac{1}{|\phi_0|\epsilon}\right)$, and the claim follows. 

One can also show that the linear dependence on $1/\Delta$ is tight: if our trial state is $\ket\phi = \phi_0\ket{\lambda_0} + \phi_1\ket{\lambda_1}$ with $\lambda_1=\lambda_0+\Delta$, then \eqref{eq:gammalowerbound2} can be replaced by 
\begin{equation} \label{eq:gammalowerbound3}
	|\gamma_{iz}| = \frac1{2^k}\left(\frac{1-e^{2\pi i r(2^k\lambda_1-z)}}{1-e^{2\pi i (\lambda_1-z/2^k)}}\right)
	\geq \frac c{\pi 2^{k-1}\Delta},
\end{equation}
which follows from Eq.~\eqref{eq:thetabound}. Thus,
$2^k = \Omega\left(\frac1{|\phi_0|\epsilon\Delta}\right)$, and the claim follows. 

\section{Filtering method by Poulin \& Wocjan} \label{app:PoulinWocjan}

Previously, Poulin and Wocjan proposed a filtering method  \cite{PoulinWocjan09} as an improvement to phase estimation that only has an inverse dependence on the overlap $|\phi_0|$. We briefly review this algorithm here, and show that it can be  adapted to achieve a runtime that scales only polylogarithmically in the allowed error (the analysis provided in \cite{PoulinWocjan09} only yields a state with low expected energy, and only an error analysis for the expected energy rather than the residual state is given there). We also show that this method can be combined with the mimimum label finding algorithm in case the ground energy is not known beforehand to first find the ground energy to the required precision.

The Poulin-Wocjan algorithm is based on the following idea:  let $\mathcal A $ be the circuit of phase estimation with $k$ ancilla qubits, but without the inverse Fourier transform. Then $\mathcal A\ket{\lambda_i}\ket0^{\otimes k} = \ket{\lambda_i}\ket{\varphi_i}$, where
\begin{equation}
 \ket{\varphi_i} := \frac1{\sqrt{2^k}}\sum_{x=0}^{2^k-1} e^{2\pi i \lambda_i x}\ket x
\end{equation}
is a momentum state encoding of $\lambda_i$. 
Since $\mathcal A$ maps $\ket{\lambda_i}\ket0^{\otimes k}$ to $\ket{\lambda_i}\ket{\varphi_i}$, then for any state $\ket\mu$ on $k$ qubits, 
$\mathcal A^\dagger$ maps $\ket\phi\ket \mu$ to 
\begin{equation}
	\sum_i \phi_i \braket{\varphi_i}{\mu} \ket{\lambda_i}\ket{0}^{\otimes k} + \ket{R}
\end{equation}
where  $\ket R $ has no overlap with $\ket{0}^{\otimes k}$ on the ancillas. Hence, starting with $\eta$ copies of $\ket\mu$ on $\eta k$ ancillas ($\eta$ will be specified later), applying $\eta$ copies of $\mathcal A^\dagger$ maps $\ket\phi\ket \mu^{\otimes\eta}$ to 
\begin{equation} \label{eq:mPWresult}
		\sum_i \phi_i \braket{\varphi_i}{\mu}^\eta \ket{\lambda_i}\ket{0}^{\otimes \eta k} + \ket{R'},
\end{equation}
where $\ket{R'}$ has no overlap with $\ket{0}^{\otimes \eta k}$ on the ancillas. The central idea of \cite{PoulinWocjan09} is that for suitable choices of $\eta$ and $\ket\mu$, $|\braket{\varphi_i}\mu|^\eta$ is a ``filter function'' that is centered around $\lambda_0$ and falls off quickly, thus suppressing all terms in \eqref{eq:mPWresult} except for the contribution from $\ket{\lambda_0}$. We now show that this idea can be used to obtain a ground state preparation algorithm whose runtime also only scales polylogarithmically with $\epsilon$. 

\begin{proposition}[Ground state preparation with filtering method for known ground energy]\label{prop:filteringKnown}
Suppose that the value of $\lambda_0$ is known to a precision of 
\begin{equation}
O\left( \frac\Delta{\sqrt{\log^{3} \frac1{|\phi_0|\epsilon}\log\log\frac1{|\phi_0|\epsilon}}}\right). 
\end{equation}
Then the Filtering method can prepare a state $\epsilon$-close to $\ket{\lambda_0}$ in a gate complexity of 
\begin{equation}\label{eq:mPWknownGates}
	\tilde O\left(\frac{\Lambda}{|\phi_0| \Delta}+\frac\Phi{|\phi_0|}\right)
\end{equation}
and using 
\begin{equation} \label{eq:filteringQubits}
	O\left(\log N + \log \frac1{\epsilon}+\frac{\log\frac1{\chi\epsilon}}{\log\log\frac{1}{\chi\epsilon}}\times\log\frac 1\Delta\right)
\end{equation}
qubits. 
\end{proposition}
\begin{proof}
Suppose  that we know a value of $\mu\in[0,1)$ such that  
\begin{equation} \label{eq:muCondition}
	|\mu-\lambda_0| < \frac1{2^{k+1}\pi\sqrt{\eta}},
\end{equation}
where $k$ and $\eta$ will be specified later. Note that this is the same as assuming that we know $\lambda_0$ up to $k+l$ binary digits, where $l:=\lceil\log_2(2\pi\sqrt \eta)\rceil$, and we can w.l.o.g. assume that $2^{k+l}\mu\in\mathbb Z$. 
Choose
\begin{equation}
	\ket\mu :=  \frac1{\sqrt{2^k}}\sum_{x=0}^{2^k-1} e^{2\pi i \mu x}\ket x. 
\end{equation}
Note that $\ket\mu$ can be efficiently prepared from the computational basis state $\ket{2^{k+l}\mu}$ by first applying the quantum Fourier transform on $k+l$ qubits, then applying Hadamard gates on the last $l$ qubits, and finally discarding the last $l$ qubits. The circuit $\mathcal C_\mu$ preparing $\ket\mu$ from $\ket{0}^{\otimes (k+l)}$ thus only requires $O((k+l)^2)$ gates \cite{NielsenChuang}. Similarly to \eqref{eq:upperboundgamma} and \eqref{eq:gammagap}, for $i\neq 0$,
\begin{equation}
	|\!\braket{\varphi_i}\mu\!|\leq \frac1{2^{k+1}|\lambda_i-\mu|} \leq \frac{1}{2^{k+1}\Delta}.
\end{equation}
Moreover, it can be shown \cite{PoulinWocjan09} that \eqref{eq:muCondition} implies
\begin{equation}
	\label{eq:lowerboundbraketeta}
	|\!\braket{\varphi_0}\mu\!|^\eta \geq \frac12. 
\end{equation}
Thus, using amplitude amplification or fixed point search to search for $\ket 0^{\otimes\eta k}$ on \eqref{eq:mPWresult}, we obtain with high probability the state
\begin{equation}
	\ket\sigma:=\frac{\sum_i \phi_i \braket{\varphi_i}{\mu}^\eta \ket{\lambda_i}}{\left\|\sum_i \phi_i \braket{\varphi_i}{\mu}^\eta \ket{\lambda_i}\right\|}
\end{equation}
with $O(\eta/|\phi_0|)$ uses of $\mathcal A^\dagger$ and $\mathcal C_\mu$, and $O(1/|\phi_0|)$ uses of $\mathcal C_\phi$. 

We now show that $\ket\sigma$ is a good approximation of $\ket{\lambda_0}$ for appropriate choices of $k$ and $\eta$. We have that
\begin{align}
\frac12\|\ket\sigma - \ket{\lambda_0}\|^2 &\leq \frac{\|\sum_{i\neq 0} \phi_i \braket{\varphi_i}{\mu}^\eta \ket{\lambda_i}\|^2}{\left\|\sum_i \phi_i \braket{\varphi_i}{\mu}^\eta \ket{\lambda_i}\right\|^2} \\
&= \frac{\sum_{i\neq 0} |\phi_i|^2 |\!\braket{\varphi_i}\mu\!|^{2\eta}}{\sum_{i} |\phi_i|^2 |\!\braket{\varphi_i}\mu\!|^{2\eta}} \\
&\leq 4\left(\frac{1}{2^{k+1}\Delta}\right)^{2\eta} \frac{1}{|\phi_0|^2}.
\end{align}
Thus, in order to obtain
\begin{equation}
 \|\ket\sigma - \ket{\lambda_0}\| < \epsilon,
\end{equation}
 it suffices if $k$ and $\eta$ satisfy
\begin{equation} 
k = \left\lceil \log_2\frac1\Delta\right\rceil + O\left(\log\log\frac1{|\phi_0|\epsilon}\right)
\end{equation}
 and 
\begin{equation}
  \eta = O\left(\frac{\log\frac 1{|\phi_0|\epsilon}}{\log\log \frac1{|\phi_0|\epsilon}}\right). 
\end{equation} 
But since $\eta$ is a parameter of the algorithm that needs to be chosen beforehand, and $|\phi_0|$ is unknown, we need to choose
\begin{equation}
  \eta = O\left(\frac{\log\frac 1{\chi\epsilon}}{\log\log \frac1{\chi\epsilon}}\right)
\end{equation} 
to ensure the algorithm works. 
The full algorithm thus requires 
\begin{equation}
\tilde O\left(\frac{\eta 2^k\Lambda + \Phi}{|\phi_0|}\right) = \tilde O\left(\frac{\Lambda}{|\phi_0| \Delta}+\frac\Phi{|\phi_0|}\right)
\end{equation} 
gates and 
\begin{equation}
O(\log N) + \eta k=O\left(\log N + \log \frac1{\epsilon}+\frac{\log\frac1{\chi\epsilon}}{\log\log\frac{1}{\chi\epsilon}}\times\log\frac 1\Delta\right)
\end{equation}
  qubits. 
  \end{proof} 
  There does not appear to be an obvious way, such as a recycling scheme, to reduce the number of qubits required.

 Suppose now that the value of $\mu$ is not known beforehand. We show now that in this case, one can combine the filtering method with the minimum label finding algorithm to determine a suitable value of $\mu$ beforehand.
 
 \begin{proposition}[Ground state preparation with filtering method for unknown ground energy] \label{prop:filteringUnknown}
 	If the ground energy is not known beforehand, the same task as in Proposition~\ref{prop:filteringKnown} can be achieved in a gate complexity of 
 	\begin{equation}
 		\tilde O\left(\frac\Lambda{\chi\Delta^{3/2}} + \frac\Phi{\chi\sqrt\Delta}\right)
 	\end{equation}
 	and the same number \eqref{eq:filteringQubits} of qubits. 
 \end{proposition} 
 \begin{proof}
 Let $k$ and $\eta$ be defined as in the proof of Proposition~\ref{prop:filteringKnown}. 
 Let $\mu_j = j/2^{k_1}$, where $k_1\geq k$ will be chosen later. It is easy to prepare the state 
 \begin{equation} \label{eq:filterSuperposMu}
 	\frac1{\sqrt{2^{k_1}}} \sum_{j=0}^{2^{k_1}-1} \ket j \ket{\mu_j}^{\otimes \eta_1} \ket \phi,
 \end{equation}
 where
 \begin{equation}
 \ket{\mu_j} = \frac1{\sqrt{2^{k_1}}}\sum_{x=0}^{2^k-1} e^{2\pi i \mu_j x}\ket x. 
 \end{equation}
 We  now run a controlled version of the filtering algorithm with $\eta_1\times k_1$ ancilla qubits. This produces the state
 \begin{equation}
 	 \sum_{j=0}^{2^{k_1}-1} \ket j \ket{\Phi_j},
 \end{equation}
 where
 \begin{equation}
 	\ket{\Phi_j} = \frac1{\sqrt{ 2^{k_1}}}\sum_i \phi_i \braket{\varphi_i}{\mu_j}^{\eta_1} \ket{\lambda_i}.
 \end{equation}
 Let $J$ be the smallest integer such that $|\mu_J-\lambda_0 | < \frac{1}{2^{k+2}\pi\sqrt\eta}$. Then from \eqref{eq:lowerboundbraketeta}, $\|\ket{\Phi_J}\| \geq \frac{|\phi_0|}{2\sqrt{ 2^{k_1}}} \geq \frac{\chi}{2\sqrt{ 2^{k_1}}}$. Let $\tilde J<J$ be an integer such that
 \begin{equation}
 	\sum_{j=0}^{\tilde J} \|\ket{\Phi_j}\|^2 = O\left(\frac{\chi^2}{2^{k_1}}\right). 
 \end{equation}
 Then, the minimum label finding algorithm finds an integer $j\in[\tilde J,J]$ with high probability. To obtain a  value of $j$ that gives rise to a good approximation of $\lambda_0$, we need to ensure that $\mu_J-\mu_{\tilde J} < \frac{1}{2^{k+2}\pi\sqrt\eta}$, since the latter then implies $|\mu_{j} - \lambda_0| < \frac{1}{2^{k+1}\pi\sqrt\eta}$ for all $j\in[\tilde J,J]$. 
 
 Let $D= J-\tilde J$. We have
 \begin{align}
 	\sum_{j=0}^{\tilde J} \|\ket{\Phi_j}\|^2 &= \frac1{2^{k_1}} \sum_{j=0}^{\tilde J } \sum_i |\phi_i|^2 |\braket{\varphi_i}{\mu_j}|^{2\eta_1} \\
 	&\leq  \frac1{2^{k_1}}  \sum_{j=0}^{\tilde J } \sum_i |\phi_i|^2 \frac{1}{(2^{k_1+1}|\lambda_i-\mu_j|)^{2\eta_1}} \\
 	&<  \frac1{2^{k_1}}    \sum_{j=0}^{\tilde J }\frac{1}{(2^{k_1+1}|\lambda_0-\mu_j|)^{2\eta_1}} \\
 	&<  \frac1{2^{k_1}}    \sum_{j>D }\frac{1}{j^{2\eta_1}} \\
 &< \frac1{2^{k_1}}   \frac1{D^{2\eta_1-1}}. 
 \end{align}
 We thus require 
 \begin{equation}
 \frac1{D^{2\eta_1-1}} < \chi^2,
 \end{equation}
 with $D=O(2^{k_1}\xi_F)$, where $\xi_F = 1/(2^{k+1}\pi\sqrt\eta)$ is the required precision \eqref{eq:muCondition}. Thus, it suffices to choose
 \begin{equation}
 	2^{k_1} = O\left(\frac1{\xi_F} \log\frac1\chi\right)=O\left( \frac1\Delta\log^{3/2}\frac1{\chi\epsilon}\sqrt{\log\log \frac1{\chi\epsilon}}\log\frac1\chi\right)
 \end{equation}
 and
 \begin{equation}
 	\eta_1 = O\left(\frac{\log \frac1\chi}{\log\log \frac1{\chi}}\right).
 \end{equation}
 This will provide an estimate of $\lambda_0$ that can be used for the state preparation. The number of gates for this estimation is
 \begin{equation}
  \tilde O\left(\frac{\sqrt{2^{k_1}}}{\chi}\left(2^{k_1}\Lambda + \Phi\right)\right) = \tilde O\left(\frac{\Lambda}{\chi\Delta^{3/2} }+ \frac{\Phi}{\chi\Delta^{1/2}}\right),
 \end{equation}
 as claimed.
 \end{proof}
 
 Note that in analogy to Section~\ref{subsec:proofsUnknown}, the algorithm can be used to estimate the ground energy to an arbitrary precision $\xi = \tilde O(\Delta)$, by simply running the algorithm with a smaller value of $\Delta$. 
 
 \begin{corollary}[Ground energy estimation with filtering method] \label{cor:filteringEstimate}
 	Let $\xi=\tilde O(\Delta)$. Then the Filterning method can be used to estimate $\lambda_0$ to an additive precision of $\xi$ in a gate complexity of 
 	\begin{equation}
 		\tilde O\left( \frac\Lambda{\chi\xi^{3/2} } + \frac\Phi{\chi\sqrt\xi} \right), 
 	\end{equation}
 	and using 
 	\begin{equation} \label{eq:filteringEstimateQubits}
 		O\left( \log N + \frac{\log\frac1\chi}{\log\log\frac1\chi} \times\log\frac1\xi \right)
 	\end{equation}
 	qubits. 
 \end{corollary}

 
 Alternatively, one can also use a combined approach by using a prior run of phase estimation to first get a ``crude'' estimate of $\lambda_0$, similarly to the method in Section~\ref{sec:unknown} of the main text. This approach is useful if $\Delta$ is very small. 
 
 \begin{proposition}[Combining filtering method with phase estimation]
 	By combining the Filtering method with phase estimation, 
 	\begin{enumerate}[(i)]
 		\item The same task as in Proposition~\ref{prop:filteringUnknown} can be achieved in a gate complexity of 
 		\begin{equation}
 		 	\tilde O\left(\frac{\Lambda}{\chi^3\Delta^\kappa} + \frac{\Lambda}{\chi \Delta^{(3-\kappa)/2} }+\frac{\Phi}{\chi\Delta^{(1-\kappa)/2}}\right)
 		\end{equation}
 		and the same number \eqref{eq:filteringQubits} of qubits,
 		\item The same task as in Corollary~\ref{cor:filteringEstimate} can be achieved in a gate complexity of 
 		\begin{equation}
 		 	\tilde O\left(\frac{\Lambda}{\chi^3\xi^\kappa} + \frac{\Lambda}{\chi \xi^{(3-\kappa)/2} }+\frac{\Phi}{\chi\Delta^{(1-\kappa)/2}}\right)
 		\end{equation}
 		and the same number \eqref{eq:filteringEstimateQubits} of qubits,
 	\end{enumerate}
 	where $\kappa\in[0,1]$ is arbitrary.
 \end{proposition}
 \begin{proof}
 To prove (i), first use the method from Appendix~\ref{app:PEAenergy} to obtain the ground energy to a precision of $\xi' = \Delta^\kappa$ for some $\kappa\in[0,1]$ chosen below. This provides us with an interval $[a,b]\ni \lambda_0 $ with $b-a = O(\Delta^\kappa)$. Let $\mu_j' = a+(b-a)j/L$ with $L=\Theta(\Delta^\kappa/\xi)$. Note that $\mu_{j+1}'-\mu_j' < \xi$. We now run the  algorithm from Proposition~\ref{prop:filteringUnknown}, but replacing \eqref{eq:filterSuperposMu} with
 \begin{equation}
 	\frac1{\sqrt L} \sum_{j=0}^L \ket j \ket{\mu_j'}^{\otimes\eta}\ket\phi. 
 \end{equation}
 Then, the total number of gates is
 \begin{equation}
 	\tilde O\left(\frac{\Lambda}{\chi^3\xi'} + \frac{\sqrt L}{\chi}(2^{k_1}\Lambda+\Phi)\right)
 	= \tilde O\left(\frac{\Lambda}{\chi^3\Delta^\kappa} + \frac{\Lambda}{\chi \Delta^{(3-\kappa)/2} }+\frac{\Phi}{\chi\Delta^{(1-\kappa)/2}}\right).
 \end{equation}
 This proves part (i). Part (ii) follows from the same argument as Corollary~\ref{cor:filteringEstimate}.
 \end{proof}
 Note that choosing $\kappa=1$ gives the optimal inverse scaling in $\Delta$ and $\xi$, respectively, for this algorithm. Similarly as in Section~\ref{sec:unknown} however, other values of $\kappa$ can be chosen to take the other parameters into account. 

\section{Chebyshev method}\label{app:Chebyshev}

We now show that the alternative approach of \cite{ChildsKothariSomma15} of using quantum walks and Chebyshev polynomials can be used to obtain an algorithm with essentially the same runtime. We assume in this section that $\tilde H$ has at most $d=O(\log N)$ non-zero entries in each row/column\footnote{Notice that this includes many-body Hamiltonians, as Hamiltonians consisting of $n$ terms acting on at most $k$ qubits are sparse with $d=2^kn$.}. We also assume that the spectrum of $\tilde H$ is contained in $[0, 1/2]$ for simplicity\footnote{In fact, it is sufficient to assume that the spectrum of $\tilde H$ is contained in $[0,1-\tau]$, where $\tau$ is defined in \eqref{eq:choicetauCheby}. This ensures that $H$, as defined below, has entries with modulus at most $1$}. 
Moreover, as in previous work \cite{BerryChildsKothari15,ChildsKothariSomma15,LowChuang16}, we assume that we are given quantum oracle access to the positions and values of the non-zero elements of $\tilde H$. Specifically, we assume that we are given unitaries  $\mathcal O_1, \mathcal O_2$ such that $\mathcal O_1\ket{j,l}=\ket{j,\nu(j,l)}$ and $\mathcal O_2\ket{j,k,z} = \ket{j,k,z\oplus \tilde H_{jk}}$, where $\nu(j,l)$ is the column index of the $l^{\text{th}}$ nonzero entry in the $j^{\text{th}}$ row of $\tilde H$, and the third register on which $\mathcal O_2$ acts encodes a bit string representation of an entry of $\tilde H$. In this section, let $\Lambda$ denote the gate complexity of the oracles. 

Suppose that the value of $\lambda_0$ is known to a precision of $\delta=O\left( \Delta/\log\frac1{\chi\epsilon} \right)$. 
Let $ E$ be a known real number such that $0\leq  E\leq \lambda_0$ and $\delta_E := \lambda_0-E<\delta$. Define 
 $H:=\tilde H - (E-\tau)\id$. Then $\ket{\lambda_0}$ is the unique eigenvector of $H$ with minimum eigenvalue $\tau+\delta_E$ and by assumption,  all other eigenvalues of $H$  are $\geq \tau+\delta_E+\Delta$. 
This method is based on the observation that a high power of $\id-(H/d)^2$ is approximately proportional to a projector onto $\ket{\lambda_0}$. More precisely, for any trial state $\ket\phi= \phi_0\ket{\lambda_0}+\ket{\lambda_0^\perp}$, 
\begin{equation} \label{eq:HMphi}
\left(\id - \left(\frac Hd\right)^2\right)^M\ket\phi = \phi_0\left(1-\left(\frac{\tau+\delta_E}d\right)^2\right)^M \left( \ket{\lambda_0} + \frac1{\phi_0} \left(\frac{\id - (H/d)^2}{1-((\tau+\delta_E)/d)^2}\right)^M  \ket{\lambda_0^\perp}\right).
\end{equation}
The norm of the second term in the brackets is bounded by $|\phi_0|^{-1} e^{-\Omega(M(\tau+\delta_E)\Delta)}$. 
Indeed, for small $\Delta$,
\begin{align}
	\left\|\left(\id - \left(\frac Hd\right)^2\right)^M \ket{\lambda_0^\perp}\right\| &<
	\left( \frac{1-\left(\frac{\Delta+\tau+\delta_E}d\right)^2}{1-\left(\frac{\tau+\delta_E}d\right)^2} \right)^M \\
	&= \left(1 - \frac{2\Delta(\tau+\delta_E) + \Delta^2}{d^2\left(1-\left(\frac{\tau+\delta_E}d\right)^2\right) }\right)^M \\
	&= e^{-\Omega(M(\tau+\delta_E)\Delta/d^2)}
\end{align}
Thus, 
\begin{equation} \label{eq:HMFullApprox}
	\left\| \frac{\left(\id - \left( H/d\right)^2\right)^M \ket\phi}{\left\|\left(\id - \left( H/d\right)^2\right)^M\ket\phi\right\| } - \ket{\lambda_0} \right\| =O(\epsilon),
\end{equation}
provided that 
\begin{equation}\label{eq:choiceMCheby}
M=\Omega\left(\frac{d^2}{\Delta(\tau+\delta_E)}\log\frac1{|\phi_0|\epsilon}\right).
\end{equation}
On the other hand, 
\begin{align}
	\left(1-\left(\frac{\tau+\delta_E}d\right)^2\right)^M = e^{-O( (\tau+\delta_E)^2M/d^2)}. 
\end{align}
Thus,
\begin{equation} \label{eq:chebyNorm}
\left\|\left(\id - \left(\frac Hd\right)^2\right)^M \ket\phi\right\| = \Omega(|\phi_0|),
\end{equation}
provided that $\tau+\delta_E = O(d/\sqrt M)$. Hence, since by assumption $\delta_E < \delta$, choosing
\begin{equation} \label{eq:choicetauCheby}
	\tau = \Theta\left( \frac{\Delta}{\log\frac1{\chi\varepsilon} }\right) 
\end{equation}
and 
\begin{equation}
 M = \Theta\left( \frac{d^2}{\Delta^2} \log^2\frac1{\chi\epsilon}\right)
\end{equation}
satisfies both \eqref{eq:HMFullApprox} and \eqref{eq:chebyNorm}. 

Our aim in the following is to prepare $(\id-(H/d)^2)^M\ket\phi$. The strategy we employ is as follows. First, we approximate $(\id-(H/d)^2)^M$ as a linear combination of low order Chebyshev polynomials in $H/d$. Second, we implement this linear combination with some amplitude using quantum walks and the non-unitary LCU Lemma \cite[Lemma 7]{ChildsKothariSomma15}. Third, we use amplitude amplification or fixed point search to boost the overlap with the target state.

In the following, assume for simplicity that $M=2m$ is even (the algorithm can also be adapted to odd $M$ with minor modifications). Let $\T_k(x)$ and $\U_k(x)$ be the $k$\textsuperscript{th} Chebyshev polynomials of the first and second kind, respectively. It is well-known that
\begin{equation}
(1-x^2)^M = \sum_{k=0}^M \alpha_k \T_{2k}(x),
\end{equation}
where
\begin{equation}
	\alpha_k = 
	\begin{cases}
		2^{1-2M}\binom{2M}{M} - \frac{(2M-1)!!}{2^MM!}, & k=0\\
		 (-1)^k 2^{1-2M}\binom{2M}{M+k}, & k\geq 1.
	\end{cases}
\end{equation}
Since $\|H\|\leq 1$ and $|\T_k(x)| \leq 1$ for $|x|\leq 1$, \eqref{eq:chernoff} implies that 
\begin{equation}
\left(\id-\left(\frac Hd\right)^2\right)^{2m} =  \sum_{k=0}^{m_0} \alpha_k \T_{2k}\left(\frac Hd\right)+O(\chi\epsilon), \label{eq:truncCheby}
\end{equation}
provided that  
\begin{equation} \label{eq:choicem0}
m_0 = \Theta\left(\sqrt{M\log\frac1{\chi\epsilon}}\right)
=  \Theta\left(\frac d{ \Delta}\log^{3/2}\frac1{\chi\epsilon}\right).
\end{equation}

Next, recall that $\T_k(H/d)$ can be implemented with some amplitude using quantum walks (c.f. \cite[Section 4.1]{ChildsKothariSomma15}) on a larger Hilbert space, which is obtained by first adding an ancilla qubit and then doubling the entire system: For $j\in[N]:=\{1,\ldots,N\}$, define $\ket{\psi_j}\in \cc^{2N}\otimes\cc^{2N}$ as
\begin{equation}
\ket{\psi_j} := \ket {j0} \otimes \frac{1}{\sqrt d}\sum_{\substack{l\in [N]\\ H_{jl}\neq 0}} \ket l \left( \sqrt{H_{jl}^*}\ket 0 + \sqrt{1- |H_{jl}|}\ket 1\right) 
\end{equation}
and
\begin{equation}
T := \sum_{j\in [N]} \ketbra{\psi_j} j. 
\end{equation}
Note that  the entries of $H$ have modulus at most $1$. 
Let $S$ be the swap operator on $\cc^{2N}\otimes\cc^{2N}$, i.e. $S\ket{jb_1}\ket{lb_2}=\ket{lb_2}\ket{jb_1}$, and $W=S(2TT^\dagger - \id)$. By \cite[Lemma 16]{ChildsKothariSomma15}, within the invariant subspace $\vspan\{T\ket j, ST\ket j : j\in[N]\}$, $W^k$ has the form
\begin{equation}
 W^k = \left( \begin{array}{cc}
 	\T_k(H/d) & -\sqrt{1-(H/d)^2} \U_{k-1}(H/d)\\
 	\sqrt{1-(H/d)^2} \U_{k-1}(H/d) & \T_k(H/d)
 \end{array} \right),
\end{equation}
where the first block corresponds to the space $\vspan\{T\ket j : j\in [N]\}$. Hence, using $k$ steps of the walk, we can implement the transformation 
\begin{equation}
W_k\ket 0^{\otimes q}  \ket\phi = \ket 0^{\otimes q} \T_k(H/d)\ket\phi + \ket{R_k'},
\end{equation}
where 
$q:= \lceil \log_2 N\rceil + 3$ and 
$(\ketbra00^{\otimes q}\otimes \id)\ket{R_k'}=0$. 

To implement the RHS of \eqref{eq:truncCheby}, we employ the non-unitary LCU Lemma: Let $B$ be a circuit on $b=\lceil \log_2(m_0+1)\rceil$ qubits that maps $\ket 0^{\otimes b} $ to 
\begin{equation}
B\ket 0^{\otimes b} := \frac1{\sqrt\alpha} \sum_{k=0}^{m_0} \sqrt{\alpha_k}\ket k,
\end{equation}
where $\alpha = \sum_{k=0}^{m_0} \alpha_k$. Let $U = \sum_{k=0}^{m_0} \ketbra kk \otimes W_{2k}$ be the controlled quantum walk. Then,
\begin{align} 
 (B^\dagger\otimes \id) U(B\otimes\id) \ket0^{\otimes (b+q)}\ket\phi 
   &= \frac1\alpha \ket0^{\otimes b} \sum_{k=0}^{m_0} \alpha_k (\ket 0^{\otimes q}\T_{2k}(H/d)\ket\phi+\ket{R_{2k}'}) + \ket{R'}\\
   &= \frac1\alpha \ket0^{\otimes (b+q)} \sum_{k=0}^{m_0} \alpha_k \T_{2k}(H/d)\ket\phi + \ket{R}, 
\end{align} 
where $(\ketbra00^{\otimes b}\otimes \id)\ket{ R'}=0$ and $(\ketbra00^{\otimes (b+q)}\otimes \id)\ket{ R}=0$.

The final step of the algorithm is to boost the overlap with amplitude amplification or fixed point search. Note that amplitude amplification can be used without prior knowledge of the overlap \cite{BHMT02}. Alternatively, fixed point search \cite{YoderLowChuang14} can be used for this step. Measuring the ancillas will then project the state onto 
\begin{equation}
\ket{\lambda_0'}:=\frac{\sum_{k=0}^{m_0} \alpha_k \T_{2k}(H/d)\ket\phi}{\|\sum_{k=0}^{m_0} \alpha_k \T_{2k}(H/d)\ket\phi\|},
\end{equation}
provided we successfully obtain $\ket 0^{\otimes (b+q)}$ on the ancillas.  From \eqref{eq:truncCheby},
\begin{equation}
\ket{\lambda_0'} = \frac{(\id-(H/d)^2)^{2m}\ket\phi}{\|(\id-(H/d)^2)^{2m}\ket\phi\|} + O(\epsilon)
\end{equation}
thus, \eqref{eq:HMFullApprox} implies
\begin{equation}
	\ket{\lambda_0'} = \ket{\lambda_0}+O(\epsilon),
\end{equation}
as required. The probability of success is close to $1$, provided that the number of repetitions is 
\begin{equation}
O\left(\frac{\alpha}{\|\sum_{k=0}^{m_0} \alpha_k \T_{2k}(H/d)\ket\phi\|} \right) 
= O(\alpha/|\phi_0|),
\end{equation}
where the last equation follows from \eqref{eq:chebyNorm}. 

We now calculate the gate count of the entire algorithm. First note that $B$ 
can be implemented with $O(2^b)=O(m_0)$ elementary gates \cite{SBM06}. Next, note that the oracle to $H$ can be obtained from the oracle to $ \tilde H$ with $O(\log M)$ additional gates and qubits. The gate cost to implement $W$ to accuracy $\epsilon'$ is $O(\Lambda+\log M + \log N + \log^{5/2}(1/\epsilon'))$ \cite{BerryChildsKothari15}. 
Here, we require $\epsilon' = O( \epsilon|\phi_0|/m_0d)$. 
Thus, the gate cost of $U$ is $O(m_0(\Lambda+\log M+\log N + \log^{5/2}(m_0d/\epsilon|\phi_0|))$ \cite[Lemma 8]{ChildsKothariSomma15}. Note that $\alpha=O(1)$, and  each iteration of amplitude amplification or fixed point search requires $ O(1)$ uses of $\mathcal C_\phi$, $B$ and $U$. The final gate complexity is thus
\begin{equation}
O\left(\frac{1}{|\phi_0|}\left(m_0\left(\Lambda+\log M+\log N + \log^{5/2}\frac{m_0d}{\epsilon|\phi_0|}\right)+\Phi\right)\right) 
=O\left(\frac\Lambda{|\phi_0| \Delta}\polylog\left(N,\frac1\Delta,\frac1{|\phi_0|\epsilon}\right)+\frac {\Phi}{|\phi_0|}\right).
\end{equation}
The total number of qubits required is $O(\log N + \log M+\log m_0)$. \qed 

It is moreover easy to see that, analogously to Section~\ref{subsec:proofsUnknown}, this approach can also be used for ground state preparation in the case of unknown ground energy, and for estimating the ground energy.

\end{document}